\documentclass[conference]{IEEEtran}
\usepackage{enumerate}
\usepackage{graphicx}
\usepackage{mathrsfs}
\usepackage{cite}
\usepackage{algorithm}
\usepackage{algorithmic}
\usepackage{amsmath}
\usepackage{bm}
\usepackage{graphics}
\usepackage{epstopdf}
\usepackage{color}
\usepackage{amsfonts}
\usepackage{subfigure}
\usepackage{amsmath,amssymb}
\newtheorem{mythe}{Theorem}

\newtheorem{mydef}{Definition}
\newtheorem{myrem}{Remark}
\usepackage{algorithm}
\usepackage{algorithmic}
\usepackage{epstopdf}
\usepackage{fancyhdr,array}
\usepackage{caption}
\usepackage{enumerate}
\usepackage{bm}
\usepackage{amsmath,amsfonts,graphicx}
\usepackage{algorithm}
\usepackage{mathrsfs}
\usepackage{algorithmic}
\usepackage{epstopdf}
\usepackage{amssymb}
\usepackage{subfigure}

\usepackage{cite}
\usepackage{color}
\usepackage{amsfonts}
\usepackage{algorithm}
\usepackage{algorithmic}
\usepackage{amsmath}

\usepackage{makecell}
\usepackage{hyperref}
\makeatletter

\def\hlinew#1{%
  \noalign{\ifnum0=`}\fi\hrule \@height #1 \futurelet
   \reserved@a\@xhline}
\newcommand{\new}{\textcolor{black}}
\newcommand{\newnew}{\textcolor{black}}
\newcommand{\newnewnew}{\textcolor{black}}
\newcommand{\final}{\textcolor{black}}
\newcommand{\system}{\textrm{LinkMirage}}
\newcommand{\facebook}{\textrm{870K}}

\begin{document}
\title{\system{}: Enabling Privacy-preserving Analytics on Social Relationships}
\vspace{-2em}
\author{Changchang Liu, Prateek Mittal\\ Email: cl12@princeton.edu, pmittal@princeton.edu\\ Department of Electrical Engineering, Princeton University}
\IEEEoverridecommandlockouts
\makeatletter\def\@IEEEpubidpullup{9\baselineskip}\makeatother
\IEEEpubid{\parbox{\columnwidth}{Permission to freely reproduce all or part
    of this paper for noncommercial purposes is granted provided that
    copies bear this notice and the full citation on the first
    page. Reproduction for commercial purposes is strictly prohibited
    without the prior written consent of the Internet Society, the
    first-named author (for reproduction of an entire paper only), and
    the author's employer if the paper was prepared within the scope
    of employment.  \\
    NDSS '15, 8-11 February 2015, San Diego, CA, USA\\
    Copyright 2015 Internet Society, ISBN TBD\\
    http://dx.doi.org/10.14722/ndss.2015.23xxx
}
\hspace{\columnsep}\makebox[\columnwidth]{}}
\maketitle
Social relationships present a critical foundation for many real-world applications.
However, both users and online social network (OSN) providers are hesitant to share social
relationships with untrusted external applications due to privacy concerns.
In this work, we design LinkMirage, a system that mediates privacy-preserving access to
social relationships. LinkMirage takes users' social relationship graph as an input,
obfuscates the social graph topology, and provides untrusted external applications with 
an obfuscated view of the social relationship graph while preserving graph utility. \\
\indent Our key contributions are (1) a novel algorithm for obfuscating social relationship 
graph while preserving graph utility, (2) theoretical and experimental analysis 
of privacy and utility using real-world social network topologies, including a large-scale 
Google+ dataset with 940 million links. Our experimental results demonstrate that LinkMirage 
provides up to 10x improvement in privacy guarantees compared to the state-of-the-art approaches. 
Overall, LinkMirage enables the design of real-world applications such as recommendation systems, graph analytics, 
anonymous communications, and Sybil defenses while protecting the privacy of social relationships.
\section{Introduction}
Online social networks (OSNs) have revolutionized the way our society 
interacts and communicates with each other. Under the hood, OSNs can be viewed as a special graph 
structure composed of individuals (or organizations) and connections between these entities. 
These \emph{social relationships} represent sensitive relationships between entities, for example, 
trusted friendships or important interactions in Facebook, Twitter, or Google+, which users 
want to preserve the security and privacy of.\\
\indent At the same time, an increasing number of third party applications rely on users' social relationships 
(these applications can be external to the OSN). \new{E-commerce applications can leverage social relationships for improving 
sales \cite{kim:PICEC07}, and data-mining researchers also rely on the social relationships for functional 
analysis \cite{page:stanford99,newman:PNAS06}. Social relationships can be used 
to mitigate spam\cite{Mislove:NDSI08}. Anonymous communication systems can improve client anonymity by leveraging users' social relationships \cite{Dingledine:USENIX04,nagaraja:pet07,Mittal:piscesNDSS13}. 
State-of-the-art Sybil defenses rely on social trust relationships to detect attackers 
\cite{Yu:IEEES&P08, Danezis:NDSS09,liu2015exploiting}.}\\
\indent However, both users and the OSN providers are hesitant to share social relationships/graphs with these applications due to privacy concerns. 
For instance, a majority of users are exercising privacy controls provided by popular OSNs such as Facebook, Google+ and LinkedIn to limit access to their social relationships~\cite{dey:sesoc12}. 
Privacy concerns arise because external applications that rely on users' social relationships can either explicitly
 reveal this information to an 
adversary, or allow the adversary to 
perform inference attacks~\cite{gong:IMC12,liben:JASIST07,Narayaran:S&P09,srivatsa:CCS12,nilizadeh:CCS14,ji:CCS14}. These concerns hinder the deployment of many real-world applications. \final{Thus, there exist fundamentally conflicting requirements for any link obfuscation mechanism: protecting privacy for the sensitive links in social networks and preserving utility
of the obfuscated graph for use in real-world applications.}\\
\indent In this work, we design \system{}, a system that mediates privacy-preserving access to 
social relationships. \system{} takes users' social relationship graph as an input, either 
via an OSN operator or via individual user subscriptions. Next, \system{} \emph{obfuscates} 
the social graph topology to protect the privacy of users' social contacts (edge/link privacy, 
not vertex privacy). \system{} then provides external applications such as \new{graph analytics} and anonymity 
systems~\cite{Dingledine:USENIX04,nagaraja:pet07,Mittal:piscesNDSS13} with an obfuscated view of 
the social relationship graph. Thus, \system{} provides a trade-off between securing the 
confidentiality of social relationships, and enabling the design of social relationship 
based applications. \\
\indent {We present a novel obfuscation algorithm that first clusters social graphs, and then
anonymizes intra-cluster links and inter-cluster links, respectively. We obfuscate links in a 
manner that preserves the key structural properties of social graphs. While our approach is of 
interest even for static social graphs, we go a step further 
in this paper, and consider the evolutionary dynamics of social graphs (node/link addition or deletion).
We design \system{} to be resilient to such evolutionary dynamics, by \emph{consistently} 
clustering social graphs across time instances. Consistent clustering improves both the privacy 
and utility of the obfuscated graphs.} 
\new{We show that \system{} provides strong privacy properties. Even a strategic adversary with full 
access to the obfuscated graph and prior information about the original social graph is limited 
in its ability to infer information about users' social relationships. \system{} provides up to 3x 
privacy improvement in static settings, and up to 10x privacy improvement in dynamic settings compared to 
the state-of-the-art approaches.}   \\
\indent Overall, our work makes the following contributions.
\vspace{-1em}
\begin{enumerate}[$\bullet$]
\item{First, we design \system{} to mediate privacy-preserving access to users' social relationships. \system{} obfuscates 
links in the social graph (link privacy) and provides untrusted external applications with an obfuscated view of the social graph. \system{} can achieve a good
balance between privacy and utility, under the context of both static and dynamic social network topologies.}\\
\item{
Second, \system{} provides rigorous privacy guarantees to defend against strategic adversaries with prior information of the 
social graph. We perform link privacy analysis both theoretically as well as using real-world social network
topologies. The experimental results for both a \emph{Facebook} dataset (with \facebook{} links) and a 
large-scale \emph{Google+} dataset (with 940M links) show up to 10x improvement in privacy over the state-of-the-art research.}\\
\item{
Third, we experimentally demonstrate the applicability of \system{} in real-world applications,
such as privacy-preserving graph analytics, anonymous communication and Sybil defenses.  
\system{} enables the design of social relationships based systems while simultaneously protecting the
privacy of users' social relationships.
}\\
\item{
Finally, we quantify a general utility metric for \system{}. We analyze
our utility measurement provided by \system{} both theoretically and using real-world social graphs (\emph{Facebook} and \emph{Google+}).
}
\end{enumerate}

\vspace{-1em}
\section{{Background}}
\subsection{Motivating Applications}\label{back_app}
{In this paper, we} focus our research
on protecting the \emph{link} privacy between \emph{labeled} 
vertices in social networks~\cite{Hay:csfacultypublication07,Mittal:NDSS13,Ying:SIAM08}. Mechanisms for graph analytics, anonymous communication, and Sybil defenses can leverage users' social relationships 
for enhancing security, but end up revealing users' social relationships to adversaries. 
For example, in the Tor network~\cite{Dingledine:USENIX04}, the relays' IP addresses (labels) are already publicly known 
(vertex privacy in~\cite{zhou:SIGKDD08,sala:imc11,liu:sigmod08} is not useful). Tor operators are 
hesitant to utilize social trusts to set up the Tor circuit as recommended by~\cite{nagaraja:pet07, Mittal:piscesNDSS13} since the circuit construction protocol would reveal sensitive social contact information about the users. Our proposed link-privacy techniques can thus be utilized by the Tor relay operators to enhance system security while preserving link privacy. Overall, our work focuses 
on protecting users' trust relationships while enabling the design of such systems. \\ 
\indent\system{} supports three categories of social relationship based applications: 
1) Global access to the obfuscated graph: Applications such as social network based anonymity systems~\cite{Dingledine:USENIX04,Mittal:piscesNDSS13,nagaraja:pet07} and peer-to-peer networks~\cite{Danezis:NDSS09}  
can utilize \system{} (described in Section~\ref{deploy}) to obtain a global view of privacy-preserving social graph topologies; 
2) Local access to the obfuscated graph: an individual user can query \system{} for his/her 
obfuscated social relationships (local neighborhood information), to facilitate distributed 
applications such as SybilLimit~\cite{Yu:IEEES&P08}; 
3) Mediated data analytics: \system{} can enable privacy-preserving data analytics by running 
desired functional queries (such as computing graph modularity and pagerank) on the obfuscated graph topology and only 
returning the result of the query. Existing work \cite{dwork:Springer06,dwork:ACM09} demonstrated that the 
implementation of graph analytics algorithms could leak certain information. Instead of repeatedly adding perturbations 
to the output of each graph analytics algorithm as in differential privacy \cite{dwork:Springer06,dwork:ACM09}, which would be rather costly, \system{} can 
obtain the perturbed graph just once to support multiple graph analytics. Such an approach protects the privacy of users' social 
relationships from inference attacks using query results. 
\indent \newnew{There exists a plethora of attacks against vertex 
anonymity based mechanisms \cite{Narayaran:S&P09,srivatsa:CCS12,nilizadeh:CCS14,ji:CCS14}. Ji et al.~\cite{Ji:Usenix15} recently showed that no single vertex anonymization technique was able to resist all the existing attacks.
Note that these attacks are not applicable to link privacy schemes. Therefore, a sound approach to vertex anonymity must start with improvements in our understanding of link privacy. {When used as first step} in the design of vertex privacy mechanisms, our approach can protect the privacy of social contacts and graph links even when the vertices are de-anonymized using state-of-the-art approaches \cite{Narayaran:S&P09,srivatsa:CCS12,nilizadeh:CCS14,ji:CCS14}. {Furthermore, 
our method can even improve the resilience of vertex anonymity mechanisms against de-anonymization attacks when applied 
to \emph{unlabelled} graphs (will be shown in Section~\ref{vertex})}}. \\
{
\begin{figure}[!t]
\renewcommand{\captionfont}{\footnotesize}
\centering
\includegraphics[width=3.3in,height=1.2in]{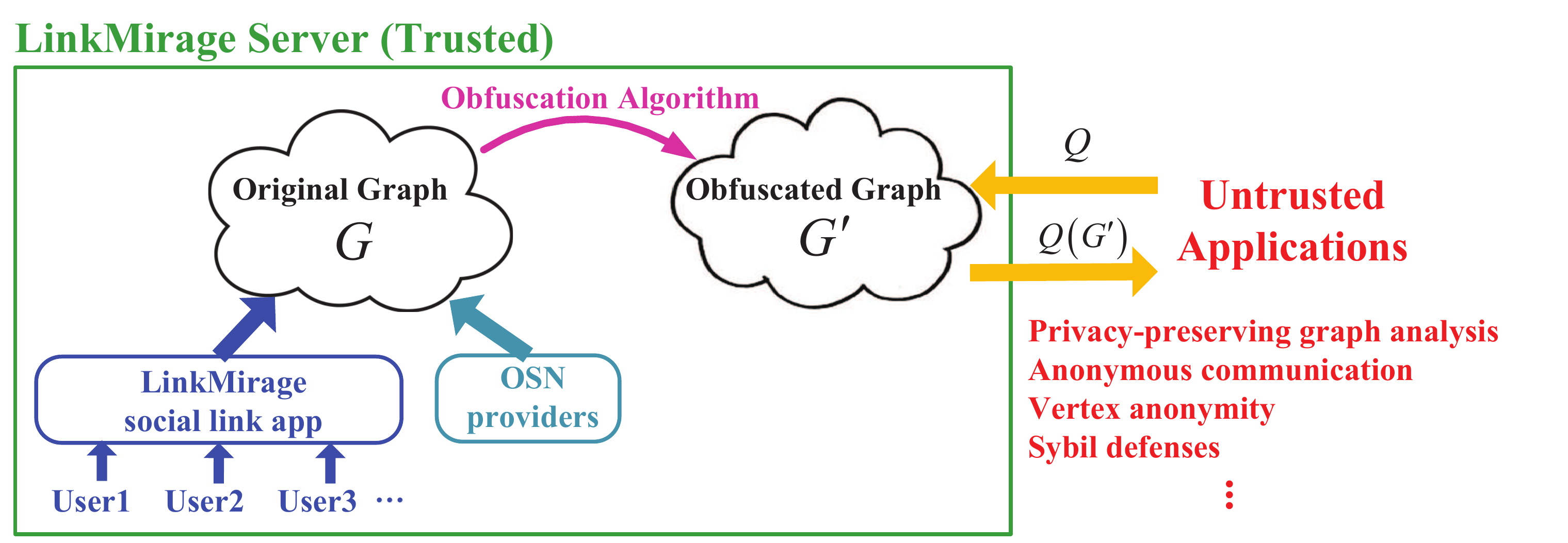}
\caption{{\system{} architecture. \system{} first collects social link information through our social link app or directly 
through the OSN providers, and then applies an obfuscation algorithm to perturb the 
original social graph(s). The obfuscated graph(s) would be utilized to answer the query 
of the untrusted applications in a privacy-preserving manner. The third-party application 
(which queries the social link information) is considered an adversary which aims to obtain 
sensitive link information from the perturbed query results. 
}}
\label{archi}
\end{figure}
\vspace{-1em}
}
\subsection{{System Architecture and Threat Model}}
\label{section_threat}
{Fig.~\ref{archi} shows the overall architecture for \system{}.
For link privacy, we consider the third-party applications (which can query the social link information) 
as adversaries, which aim to obtain sensitive link information from the perturbed query results. 
A sophisticated adversary may have access to certain prior information such as partial link information 
of the original social networks, and such prior information can be extracted from publicly available sources,  
social networks such as Facebook, or other application-related sources as stated in~\cite{burattin:arxiv14}. 
The adversary may leverage Bayesian inference to infer the probability for the existence of a link.
We assume that \system{} itself is trusted, in addition to the social network providers/users who 
provide the input social graph.} \\
\indent In Section~\ref{antiinfer}, \ref{indis}, we define our {Bayesian} privacy metric (called \emph{anti-inference privacy}) and 
an information theoretic metric (called \emph{indistinguishability}) to characterize the privacy 
offered by \system{} 
against adversaries with prior information. In addition, the evolving 
social topologies introduce another serious threat where sophisticated adversaries
can combine information available in multiple {query results} to
infer users' social relationships. We define \emph{anti-aggregation privacy} in Section~\ref{antiaggre}, 
for evaluating the privacy performance of \system{} 
against such adversaries.
\subsection{Basic Theory}
Let us denote a time series of social graphs as $G_0, \cdots, G_T$.
For each temporal graph $G_t=(V_t, E_t)$, the set of vertices is $V_t$ and the
set of edges is $E_t$. For our theoretical analysis, we focus on undirected graphs where all
the $|E_t|$ edges are symmetric, i.e. $(i,j)\in E_t$ iff $(j,i)\in E_t$.
{Note that our approach can be generalized to directed graphs with asymmetric edges.}
$P_t$ is the transition probability matrix of the Markov chain on the vertices of $G_t$. $P_t$ measures
the probability that we follow an edge from one vertex to another vertex,
where $P_t\left(i,j\right)=1/{\deg(i)}$ ($\deg(i)$ denotes the degree of vertex $i$) if $(i,j)\in E_t$,
otherwise $P_t\left(i,j\right)=0$. \new{A random walk starting from vertex $v$, selects a neighbor of $v$ at random according to $P_t$ and repeats the process.} 
\subsection{\final{System Overview and Roadmap}}
\final{
Our objective for \system{} is to obfuscate social relationships while balancing
privacy for users' social relationships and the usability for large-scale real-world applications (as will be stated in Section~\ref{design_goals}). We deploy \system{} as a Facebook application that implements graph construction and obfuscation (as will be discussed in Section~\ref{deploy}). We then describe the perturbation mechanism of \system{} in Section~\ref{sec:perturb} where we take both the static and the temporal social network topology into consideration. Our perturbation mechanism consists of two steps: dynamic clustering which finds
community structures in evolving graphs by simultaneously considering consecutive graphs, and selective perturbation which perturbs the minimal amount of edges in the evolving graphs. Therefore, it is possible to use a very high privacy parameter in the perturbation process, while preserving structural properties such as community structures. We then discuss the scalability of our algorithm in Section~\ref{scalable} and visually show the effectiveness of our algorithm in Section~\ref{Facebook}. In Section~\ref{privacy}, we rigorously analyze the privacy advantage of our \system{} over the state-of-the-art work, through considering three adversarial scenarios including the worst-case Bayesian adversary. In Section~\ref{app}, we apply our algorithm on various real world applications of anonymity systems, Sybil defenses and privacy-preserving analytics. In Section~\ref{utilitysec}, we further analyze the effectiveness of \system{} on preserving different kinds of graph structural performance.}
\vspace{-1em}
\section{LinkMirage System}\label{section_protocol}
\subsection{Design Goals}\label{design_goals}
We envision that applications relying on social relationships between
users can bootstrap this information from online social network
operators such as Facebook, Google+, Twitter
with access to the users' social relationships.
To enable these applications in a privacy-preserving manner, a perturbed 
social graph topology (by adding noise to the original graph topology) should be available.\\
\indent Social graphs evolve over time, and the third-party applications would
benefit from access to the most current version of the graph.
A baseline approach is to perturb each graph snapshot independently.
However, the sequence of perturbed graphs provide significantly
more observations to an adversary than just a single
perturbed graph. We argue that an effective perturbation method
should consider the evolution of the original graph
sequence. Therefore, we have the overall design goals for our system as:
\begin{enumerate}
\item{We aim to obfuscate social relationships while balancing privacy for users' social relationships and the usability for real-world applications.}
\item{We aim to handle both the static and dynamic social network topologies.}
\item{Our system should provide rigorous privacy guarantees to defend against adversaries who have prior information of the original graphs, and adversaries who can combine multiple released graphs to infer more information.}
\item{Our method should be scalable to be applied in real-world large-scale social graphs.}
\end{enumerate}
\begin{figure}[!t]
\renewcommand{\captionfont}{\footnotesize}
\centering
\includegraphics[width=2.8in,height=1.6in]{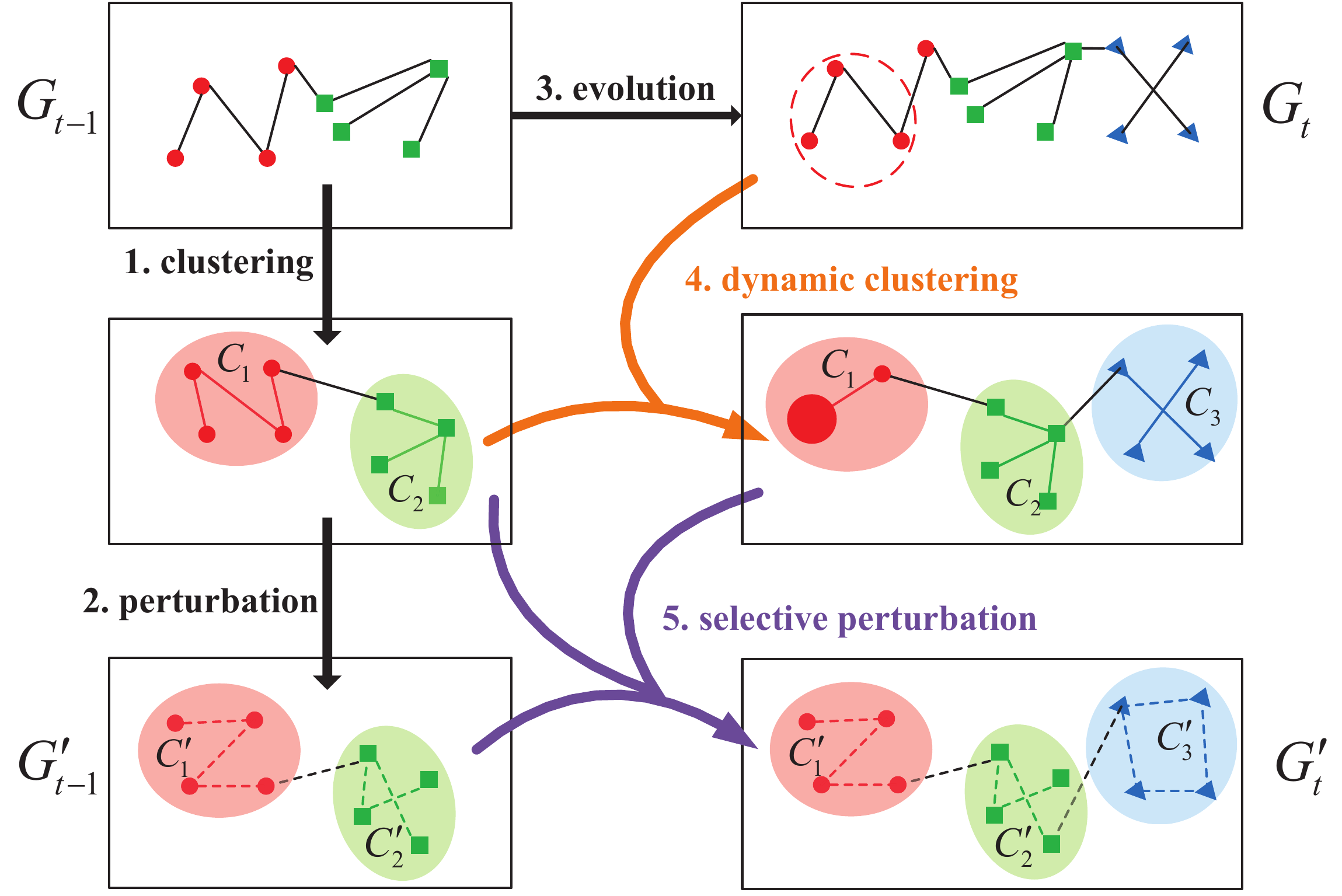}
\caption{\footnotesize{Our perturbation mechanism for $G_t$.
Assume that $G_{t-1}$ has already been dynamically
obfuscated, based on dynamic clustering (step 1) and
selective perturbation (step 2). Our mechanism
analyzes the evolved graph $G_t$ (step 3)
and dynamically clusters $G_t$ (step 4) based on the freed $m$ hop neighborhood ($m=2$)
of new links (between green and blue nodes), the merging virtual node (the large red node
in step 4), and the new nodes. By comparing the communities
in $G_{t-1}$ and $G_t$, we can implement selective
perturbation (step 5), i.e. perturb the changed blue community
independently and perturb the unchanged red and green communities in the
same way as $G_{t-1}^{\prime}$, {and then perturb the inter-cluster links.}}}
\label{mechanism}
\end{figure}
\vspace{-1em}
\subsection{LinkMirage: Deployment}\label{deploy}
To improve the usability of our proposed obfuscation approach (which will be described in detail in Section~\ref{sec:perturb}),
and to avoid dependance on the OSN providers,
we developed a Facebook application \new{(available:
\href{https://apps.facebook.com/xxxx/}
{https://apps.facebook.com/xxxx/})\footnote{Anonymized.}}
that implements graph construction (via individual user subscriptions) and obfuscation.
The work flow of the \system{} deployment is as follows:
(i) When a user visits the above URL, Facebook checks
the credentials of the user, asks whether to grant the
\emph{user's friends} permission, and then gets redirected
 to the application hosting server. (ii) The application
server authenticates itself, and then queries Facebook for
the information of the user's friends, and returns their information such as
\emph{user's id}. The list of user's friends
can then be collected by
the application server to construct
a Facebook social graph for the current {timestamp}. Leveraging \system{},
a perturbed graph {for this \emph{timestamp}} would be available which preserves
the link privacy of the users' social relationships.\\
\indent Real-world systems such as Uproxy, Lantern, Kaleidoscope \cite{hodson:Newscientist23}, 
anonymity systems~\cite{Dingledine:USENIX04,nagaraja:pet07,Mittal:piscesNDSS13}, 
Sybil defenses systems \cite{Yu:IEEES&P08,Danezis:NDSS09} can directly benefit from 
our protocol through automatically obtaining the perturbed social relationships. 
Furthermore, our protocol can enable privacy-preserving graph analytics for OSN providers. 
We will give more detailed explanations for supporting applications in Section~\ref{supportapp}.

\subsection{LinkMirage: Perturbation Algorithm}
\label{sec:perturb}
Social networks evolve with time and publishing a time series of perturbed graphs raises a
serious privacy challenge: an adversary can combine information
available from multiple perturbed graphs over time to compromise
the privacy of users' social contacts \cite{tai:ICDE11,ding:globalcom11,Bhagat:WWW10}. In \system{}, we take a time series of graph topologies into consideration, to account for the evolution of the social networks. Intuitively, the scenario with a static graph topology is just a special situation
of the temporal graph sequence, and is thus inherently incorporated in our model. \\
\indent Consider a social graph series $G_0=\left(V_0, E_0\right)$,$\cdots$,$G_T=\left(V_T, E_T\right)$.
We want to transform the graph series to $G_0^{\prime}=\left(V_0, E_0^{\prime}\right)$,$\cdots$,$G_T^{\prime}=\left(V_T, E_T^{\prime}\right)$,
such that the vertices in $G_t^{\prime}$ remain the same as in the original graph $G_t$,
but the edges are perturbed to protect link privacy. Moreover, while perturbing
the current graph $G_t$, \system{} has access to the past graphs in the
time series (i.e., $G_0, \cdots, G_{t-1}$). Our perturbation goal is to balance the utility 
of social graph topologies and the privacy of users' social contacts,
across time. 
\begin{table*}[!t]
\normalsize
\centering
\caption{{Temporal Statistics of the Facebook Dataset.}}
\begin{tabular}{c|c|c|c|c|c|c|c|c|c}
\hline
{Time} & {0} & {1} & {2} & {3} & {4} & {5}& {6}& {7} &{8} \\
\hline
{\# of nodes} & {9,586} & {9,719} & {11,649} & {13,848} & {14,210} & {16,344} & {18,974} & {26,220} & {35,048} \\
\hline
\# of edges & 48,966 & 38,058 & 47,024 & 54,787 & 49,744 & 58,099 &65,604 & 97,095&  142,274 \\
\hline
{Average degree} & {5.11} & {3.91} & {4.03} & {3.96} & {3.50} & {3.55} & {3.46} & {3.70} & {4.06}\\
\hline
\end{tabular}
\label{tablefb}
\end{table*}\\
\indent {\it Approach Overview:} Our perturbation mechanism for
\system{} is illustrated in Fig.~\ref{mechanism}.\\
\indent {{\bf{Static scenario:}} For a static graph $G_{t-1}$, we first cluster it into several communities, and then perturb the links within each community. The inter-cluster links are also perturbed to protect their privacy.}\\
\indent {{\bf{Dynamic scenario:}}} Let us suppose that $G_t$ evolves from $G_{t-1}$
by addition of new vertices (shown in blue color).
To perturb graph $G_t$, our intuition is to consider
the similarity between graphs $G_{t-1}$ and $G_t$.\\
\indent First, we partition $G_{t-1}$ and $G_t$ into subgraphs,
by clustering each graph into different communities.
To avoid randomness (guarantee consistency) in the clustering procedure and to
reduce the computation complexity, we dynamically cluster
the two graphs \emph{together} instead of clustering
them independently. Noting that one green node evolves by
connecting with a new blue node, we free \new{\footnote{We free the nodes from the previously clustering hierarchy.}} all the nodes located
within $m=2$ hops of this green node (the other two green nodes and one red node) and merge the remaining three red nodes to a big virtual node. Then, we cluster these new nodes, the freed nodes and the remaining virtual node
to detect communities in $G_t$.\\
\indent Next, we compare the communities within $G_{t-1}$ and $G_t$,
and identify the \textit{changed} and \textit{unchanged} subgraphs.
For the \textit{unchanged} subgraphs $C_1, C_2$, we set
their perturbation at time $t$ to be identical to their
perturbation at time $t-1$, denoted by
$C_1^{\prime}, C_2^{\prime}$.
For the \textit{changed} subgraph $C_3$, we perturb it
independently to obtain $C_3^{\prime}$. {We also perturb the links between communities to protect privacy of these inter-cluster links.}
Finally, we publish
 $G_t^\prime$ as the combination of $C_1^{\prime}, C_2^{\prime}, C_3^{\prime}$ and the perturbed inter-cluster links.
There are two key steps in our algorithm: dynamic clustering
and selective perturbation, which we describe in
detail as follows.
\subsubsection{{Dynamic Clustering}}
Considering that communities in social networks change significantly over time, we need to address the inconsistency problem by developing a dynamic community detection method.
Dynamic clustering aims to find community structures
in evolving graphs by simultaneously considering
consecutive graphs in its clustering algorithms.
There are several methods in the literature to
cluster evolving graphs \cite{Aynaud:Springer13},
but we found them to
be unsuitable for use in our perturbation mechanism.
One approach to dynamic clustering involves performing
community detection at each timestamp independently, and 
then establishing relationships between communities to track
their evolution \cite{Aynaud:Springer13}.
We found that this approach suffers from performance issues
 induced by inherent randomness in clustering algorithms, in
addition to the increased computational complexity.
\begin{algorithm}[!t]
\renewcommand{\algorithmicrequire}{\textbf{Input:}}
\renewcommand\algorithmicensure {\textbf{Output:} }
\begin{algorithmic}[1]
{
\REQUIRE~~{$\{G_{t}, G_{t-1}, G_{t-1}^{\prime}\}$ if $t\ge1$ or $\{G_t\}$ if $t=0$;}\\
\ENSURE  {$G_t^\prime$};\\
       {
       $G_t^{\prime},C_t=$null;\\
       \textbf{if} t=0;\\
       $~~~$cluster $G_0$ to get $C_0$;\\
       $~~~$label $C_0$ as changed, i.e. $C_{0\mathrm{-ch}}=C_0$;\\
       \new{\textbf{endif}} \\
       /*{\textit{\newnew{Begin Dynamic Clustering}}}*/\\
       1. free the nodes within $m$ hops of the changed links;\\
       2. re-cluster the new nodes, the freed nodes, the remai-\\
       $~~~$-ning merged virtual nodes in $C_{(t-1)}$ to get $C_t$;\\
       /*{\textit{\newnew{End Dynamic Clustering}}}*/\\
       \noindent/*{\textit{\newnew{Begin Selective Perturbation}}}*/\\
       3. find the unchanged communities $C_{t\mathrm{-un}}$ and the chan-\\
       $~~~$-ged communities $C_{t\mathrm{-ch}}$; \\
       4. let $G_{t\mathrm{-un}}^{\prime}=G_{(t-1)\mathrm{-un}}^{\prime}$;\\
       5. perturb $C_{t\mathrm{-ch}}$ for $G_{t\mathrm{-ch}}^{\prime}$ by the static method; \\
       6. \textbf{foreach} community pair $a$ and $b$;\\
       $~~~~~$\textbf{if} both of the communities belong to $C_{t\mathrm{-un}}$\\
       $~~~~~~~~~C_{t\mathrm{-in}}^{\prime}(a,b)=C_{(t-1)\mathrm{-in}}^{\prime}(a,b)$;\\
       $~~~~~$\textbf{else}\\
       $~~~~~~~$\textbf{foreach} marginal node $v_a(i)$ in $a$ and $v_b(j)$ in $b$\\
       $~~~~~~~~~~$randomly add an edge $(v_a(i), v_b(j))$ with pro-\\
       $~~~~~~~~~~$-bability $\frac{\deg(v_a(i))\deg(v_b(j))|v_a|}{|E_{ab}|(|v_a|+|v_b|)}$ to $G_{t\mathrm{-in}}^{\prime}(a,b)$;\\
       /*{\textit{\newnew{End Selective Perturbation}}}*/\\
       return $G_t^{\prime}=[G_{t\mathrm{-ch}}^{\prime}, G_{t\mathrm{-un}}^{\prime}, G_{t\mathrm{-in}}^{\prime}]$;\\
       }
       }
\caption{{\system{}, with dynamic clustering (steps 1-2) and selective perturbation (steps 3-6). The parameter $k$ denotes the perturbation level for each community. 
Here, {\emph{\small{ch, un, in}}} are short for \emph{changed, unchanged, inter-community}, respectively.}}
\label{dynamic_algorithm}
\end{algorithmic}
\end{algorithm}
Another approach is 
to combine multiple graphs into a single coupled graph \cite{Aynaud:Springer13}.
The coupled graph is constructed by adding edges between
the same nodes across different graphs. Clustering can be
performed on the single coupled graph.
We found that the clustering performance is very sensitive to
the weights of the added links, resulting in unstable clustering
results.
Furthermore, the large dimensionality of the coupled
graph significantly increases the computational
overhead.\\
\indent For our perturbation mechanism, we develop an adaptive dynamic clustering
approach for clustering the graph $G_t$ using the clustering result for
the previous graph $G_{t-1}$.
This enables our perturbation mechanism to
(a) exploit the link correlation/similarity in consecutive graph snapshots,
and (b) reduce computation complexity by avoiding repeated
clustering for unchanged links.\\
\indent \final{Clustering the graph $G_t$ from the clustering result of
the previous graph $G_{t-1}$ requires a \emph{backtracking} strategy.
We use the maximum-modularity method \cite{newman:PNAS06}
for clustering, which is hierarchical and thus easy to backtrack.
Our backtrack strategy is to first maintain a history of the merge
operations that led to the current clustering. When an evolution occurs,
the algorithm backtracks over the history of merge operations, in order
to incorporate the new additions and deletions in the graph.\\
\indent More concretely, if the link between node $x$ and node $y$ is changed
(added or deleted), we omit all the $m$-hop neighborhoods of $x$ and $y$ as well as $x$ and $y$
themselves from the clustering result of the previous timestamp, and then perform
re-clustering. All the new nodes, the changed nodes and their $m$-hop neighbors, and the remaining
merged nodes in the previous clustering result would be considered as basic elements for clustering $G_t$ (recall Figure~\ref{mechanism}).}\\
\indent For efficient implementation, we store the intermediate results of
the hierarchical clustering process in a data structure. Upon
link changes between $x,y$, we free the $m$-hop neighborhood of $x,y$ from the stored data structure.
\begin{figure*}[!t]
\renewcommand{\captionfont}{\footnotesize}
\centering
\includegraphics[width=6in,height=1.5in]{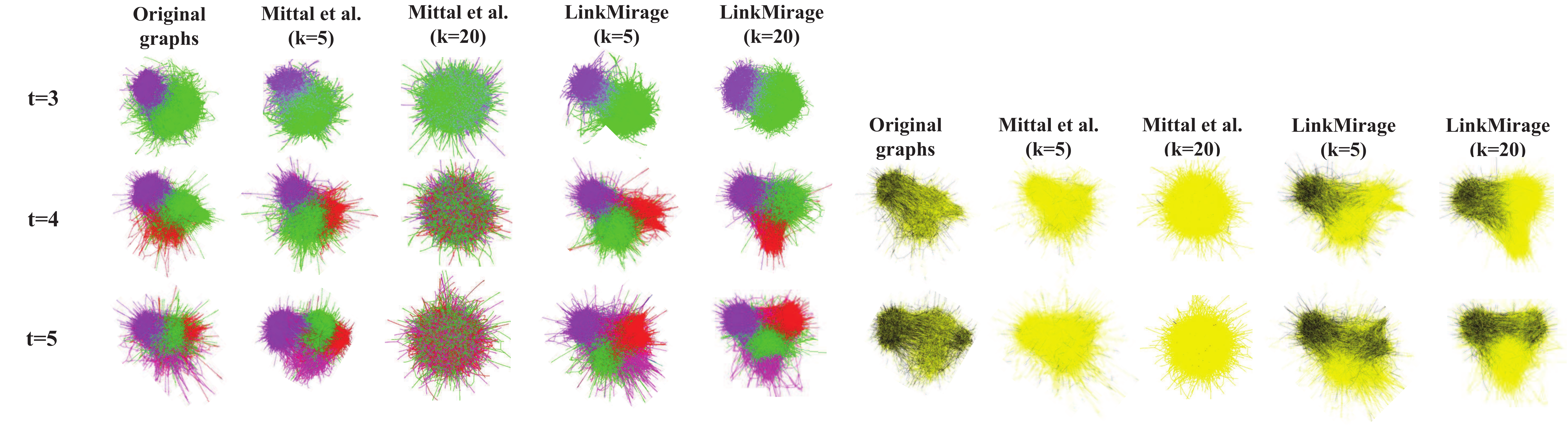}
\vspace{-1em}
\caption{Dynamic Facebook interaction dataset topology, {for 
$t=3, 4,5$. 
On the left, we can see that \system{} has  superior utility than
the baseline approach (Mittal et al.), especially for larger values of 
$k$ (due to dynamic clustering). On the right, we show the overlapped edges (black)
and the changed edges (yellow) between consecutive graphs: t=($3,4$) and t=($4,5$). 
We can see that in \system{}, the perturbation of unchanged communities is correlated across time (selective perturbation), minimizing 
information leakage and enhancing privacy. 
}}
\label{visualutility} 
\end{figure*}
\subsubsection{{Selective perturbation}}
\noindent{~~~~~~\\ \bf{Intra-cluster Perturbation:}} After clustering $G_t$ based on $G_{t-1}$ using our dynamic
clustering method, we perturb $G_t$ based on $G_{t-1}$
and the perturbed $G_{t-1}^{\prime}$. First, we compare the
communities detected in $G_{t-1}$ and $G_t$, and classify
them as \textit{changed} or \textit{unchanged}.
Our unchanged classification does not require that the
communities are exactly the same, but that the overlap
among vertices/links exceeds a threshold.
Our key idea is to keep the perturbation process for
links in the \textit{unchanged} communities to be identical to their
perturbation in the previous snapshot.
In this manner, we can preserve the privacy of these unchanged
links to the largest extent; it is easy to see that
alternate approaches would leak more information.
For the communities which are classified as changed, our approach
is to perturb their links independently of the perturbation in
the previous timestamp. 
For independent perturbations, we leverage the
static perturbation method of Mittal et al. in \cite{Mittal:NDSS13}. \new{Their static perturbation deletes all the edges in the original graph, and replaces each edge $(v,u)$ with a fake edge $(v,w)$ selected from the $k$-hop random walk starting from $v$. Larger perturbation parameter $k$ corresponds to better privacy and leads to worse utility}.\\
{\bf{Inter-cluster Perturbation:}} Finally, we need to interconnect the subgraphs 
identified above. Suppose that $|v_a|$ nodes and $|v_b|$ nodes are connecting
communities $a$ and $b$ respectively, and they construct an
inter-community subgraph. For each marginal
node $v_a(i)\in v_a$ and $v_b(j)\in v_b$ (here the marginal node in community
$a$ (resp.$b$) refers to the node that has neighbors in the other community $b$ (resp.$a$))
, we randomly connect them with probability
$\frac{\deg(v_a(i))\deg(v_b(j))|v_a|}{|E_{ab}|(|v_a|+|v_b|)}$.\footnote{This probability is set for the preservation of degree distributions as analyzed in Section~\ref{utilitysec}.} Here, all the computations for
$\deg(\cdot), |v_{\cdot}(\cdot)|, |E_{\cdot}|$ only consider the marginal nodes.
We can combine the perturbed links corresponding to
the unchanged communities, changed communities, and
inter-community subgraphs, to compute the output of our
algorithm, i.e., $G_t^{\prime}$. \\
\indent \system{} not only preserves the structural
characteristics of the original graph series,
but also protects the privacy of the users
by randomizing the original links. As compared
to prior work, our method provides stronger
privacy and utility guarantees for evolving
graphs. Detailed
procedures are stated in Algorithm.~\ref{dynamic_algorithm}.\\
\indent Surprisingly, our approach of first
isolating communities and then selectively perturbing
them provides benefits even in a static context!
This is because previous static approaches use
a single parameter to control the privacy/utility
trade-off. Thus, if we apply them to the whole graph
using high privacy parameters, it would destroy
graph utility (e.g. community structures).
On the other hand, \system{} applies perturbations
selectively to communities; thus it is possible to
use a very high privacy parameter in the perturbation
process, while preserving structural properties such as community structures.
\subsection{{Scalable Implementation}}\label{scalable}
Our algorithm relies on two key graph theoretical techniques: \emph{community detection} (serves as a foundation for the dynamic clustering step in \system{}) and \emph{random walk} (serves as a foundation for the selective perturbation step in \system{}). \new{The computational complexity for both \emph{community detection} and \emph{random walk} is $O(|E_t|)$~\cite{Aynaud:Springer13,Mittal:NDSS13} where $|E_t|$ is the number of edges in graph $G_t$, therefore the overall computational complexity of our approach is $O(|E_t|)$.} Furthermore, our algorithms are parallelizable. We adopt the GraphChi parallel framework in \cite{kyrola:OSDI12} to implement our algorithm efficiently using a commodity workstation (3.6 GHz, 24GB RAM). \new{Our parallel implementation scales to very large social networks; for example, the running time of \system{} is less than $100$ seconds for the large scale Google+ dataset (940 million links) (will be described in Section~\ref{dataset}) using our commodity workstation.}
\subsection{{Visual Depiction}}\label{Facebook}
For our experiments, we consider a real world
 \textit{Facebook social network} dataset
\cite{Viswanath:ACM_osn09} among \textit{New Orleans regional network}, spanning from September 2006 to January 2009.
Here, we utilize the wall post interaction data which
represents stronger trust relationships and comprises of 46,952 nodes (users) connected
by 876,993 edges.
We partitioned the dataset using three month intervals
to construct a total of 9 graph instances as shown
in Table~\ref{tablefb}. 
Fig.~\ref{visualutility} depicts the outcome of our
perturbation algorithm on the partitioned Facebook graph
sequence with timestamp $t=3,4,5$ (out of 9 snapshots),
 for varying perturbation parameter $k$ (perturbation parameter for
each community). For comparative analysis,
we consider a baseline approach~\cite{Mittal:NDSS13} that applies static
perturbation for each timestamp independently.
{In the dynamic clustering step of our experiments, we free the
two-hop neighborhoods of the changed nodes,  i.e. $m=2$.}\\
\indent The maximum-modularity clustering method yields two communities
for $G_3$, three communities for $G_4$, and four communities
for $G_5$. For the perturbed graphs, we use the same color
for the vertices as in the original graph and 
we can see that fine-grained structures (related to utility)
are preserved  for both algorithms under small
perturbation parameter $k$, even though links are
randomized. Even for high values of $k$, \system{} can preserve the macro-level (such as community-level)
structural characteristics of the graph.
On the other hand,
for high values of $k$, the static perturbation algorithm
results in the loss of structure properties, and appears to
resemble a random graph. Thus, our approach of first isolating
communities and applying perturbation at the level of
communities has benefits even in a static context.\\
\indent Fig.~\ref{visualutility} also shows the privacy benefits of our
perturbation algorithm for timestamps $t=4,5$. We can see
that \system{} reuses perturbed links (shown as black unchanged links)
in the unchanged communities
(one unchanged community for $t=4$ and two unchanged communities for $t=5$).
Therefore, \system{} preserves the privacy of
users' social relationships by considering correlations
among the
graph sequence, and this benefit does not come at the cost
of utility. In the following sections, we will formally
quantify the privacy and utility properties of \system{}.
\begin{table*}[!t]\small
\newcommand{\tabincell}[2]{\begin{tabular}{@{}#1@{}}#2\end{tabular}}
\renewcommand{\tabcolsep}{4pt}
\centering
\caption{{Temporal Statistics of the Google+ Dataset.}}
\begin{tabular}{c|c|c|c|c|c|c|c|c}
\hline
{Time} 
& {Jul.29} & {Aug.8} & {Aug.18} & {Aug.28} & {Sep.7} & {Sep.17} & {Sep.27} & {Oct.7}\\
\hline
{\# of nodes} 
 & {16,165,781} & {17,483,936} &{17,850,948} & {19,406,327} &{19,954,197} &{24,235,387} &{28,035,472} &{28,942,911}\\
\hline
{\# of edges} 
& {505,527,124} & {560,576,194} & {575,345,552} & {654,523,658} & {686,709,660} & {759,226,300} & {886,082,314} & {947,776,172}\\
\hline
{Average degree} 
& {31.2714} & {32.0624} & {32.2305} & {33.7273} & {34.4143} & {31.3272} & {31.6058} & {32.7464}\\
\hline
\end{tabular}
\label{tablegoogle}
\end{table*}
\vspace{-1em}
\subsection{{Supporting Applications}}\label{supportapp}
\vspace{-0.5em}
\newnewnew{As discussed in Section~\ref{back_app}, LinkMirage supports three types 
of applications:} 1) Global access to obfuscated graphs: real-world 
applications can utilize our protocol to automatically obtain the secure social graphs to enable social relationships based systems. For instance, Tor operators~\cite{Dingledine:USENIX04} (or other anonymous communication network such as Pisces in \cite{Mittal:piscesNDSS13}) can leverage the perturbed social relationships to set up the anonymous circuit; 2) Local access to the obfuscated graphs: an individual user can query our protocol for his/her perturbed friends (local neighborhood information), to implement distributed applications such as SybilLimit in \cite{Yu:IEEES&P08}; 3) Mediated data analysis: the OSN providers can also publish perturbed graphs by leveraging \system{} to facilitate privacy-preserving data-mining research, i.e., to implement graph analytics such as pagerank score~\cite{page:stanford99}, modularity~\cite{newman:PNAS06}, while mitigating disclosure of users' social relationships. Existing work in \cite{dwork:Springer06,dwork:ACM09} demonstrated that the implementation of graph analytic algorithms would leak certain information. To avoid repeatedly adding perturbations to the output of every graph analytic algorithm, which is rather costly, the OSN providers can first obtain the perturbed graphs by leveraging \system{} and then enable these graph analytics in a privacy-preserving manner. 

\vspace{-1em}
\section{Privacy Analysis}\label{privacy}
{We now address the question of understanding link privacy
of \system{}. 
{We propose three privacy metrics: \emph{anti-inference privacy, indistinguishability, anti-aggregation privacy} 
to evaluate the link privacy provided by \system{}.}  
Both theoretical analysis and experimental results with a
\emph{Facebook} dataset (\facebook{} links) and a
large-scale \emph{Google+} dataset (940M links)
show the benefits of \system{} over
previous approaches.  We also illustrate the relationship
between our privacy metric and differential privacy.}
\subsection{Experimental Datasets}\label{dataset}
To illustrate how the temporal information degrades privacy,
we consider two social network datasets. The first one is a large-scale
Google+ dataset \cite{gong:IMC12}.
whose temporal statistics are illustrated
 in Table~\ref{tablegoogle}. 
To the best of our knowledge,
this is the largest temporal dataset of social networks in public domain.
The Google+ dataset is crawled from July 2011 to October 2011 which has 28,942,911 nodes and 947,776,172 edges. 
The dataset only considers link additions, i.e. all the edges in the
previous graphs exist in the current graph. We partitioned the dataset into
84 timestamps. The second one is
the 9-timestamp Facebook wall posts dataset \cite{Viswanath:ACM_osn09} as we stated in Section~\ref{Facebook}.
with temporal characteristics shown in Table~\ref{tablefb}. {It is worth noting
that the wall-posts data experiences tremendous churn with
only 45\% overlap for consecutive graphs. Since our
dynamic perturbation method relies on the correlation between
consecutive graphs, the evaluation of our dynamic method
on the Facebook wall posts data is conservative. 
To show the improvement in performance of our algorithm for 
graphs that evolve at a slower rate, 
we also consider a sampled graph sequence extracted from the Facebook wall posts data 
with 80\% overlap for consecutive graphs.}
\begin{figure*}[!t]
\renewcommand{\captionfont}{\footnotesize}
\centering{
\label{probabilitywhole} 
\includegraphics[width=2.7in,height=1.3in]{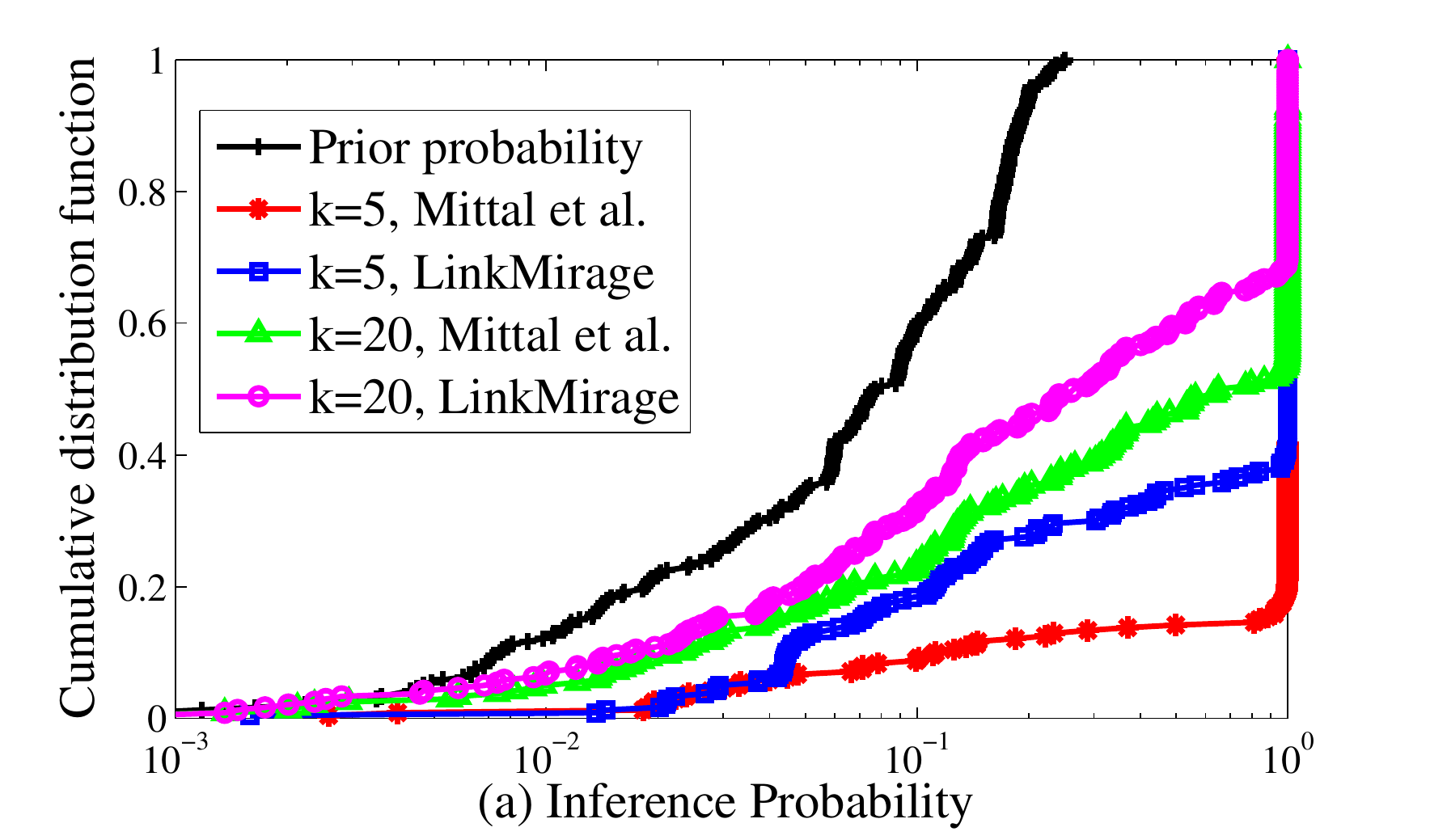}}
\hspace{0.5in}
\centering{
\label{probability80percent} 
\includegraphics[width=2.7in,height=1.3in]{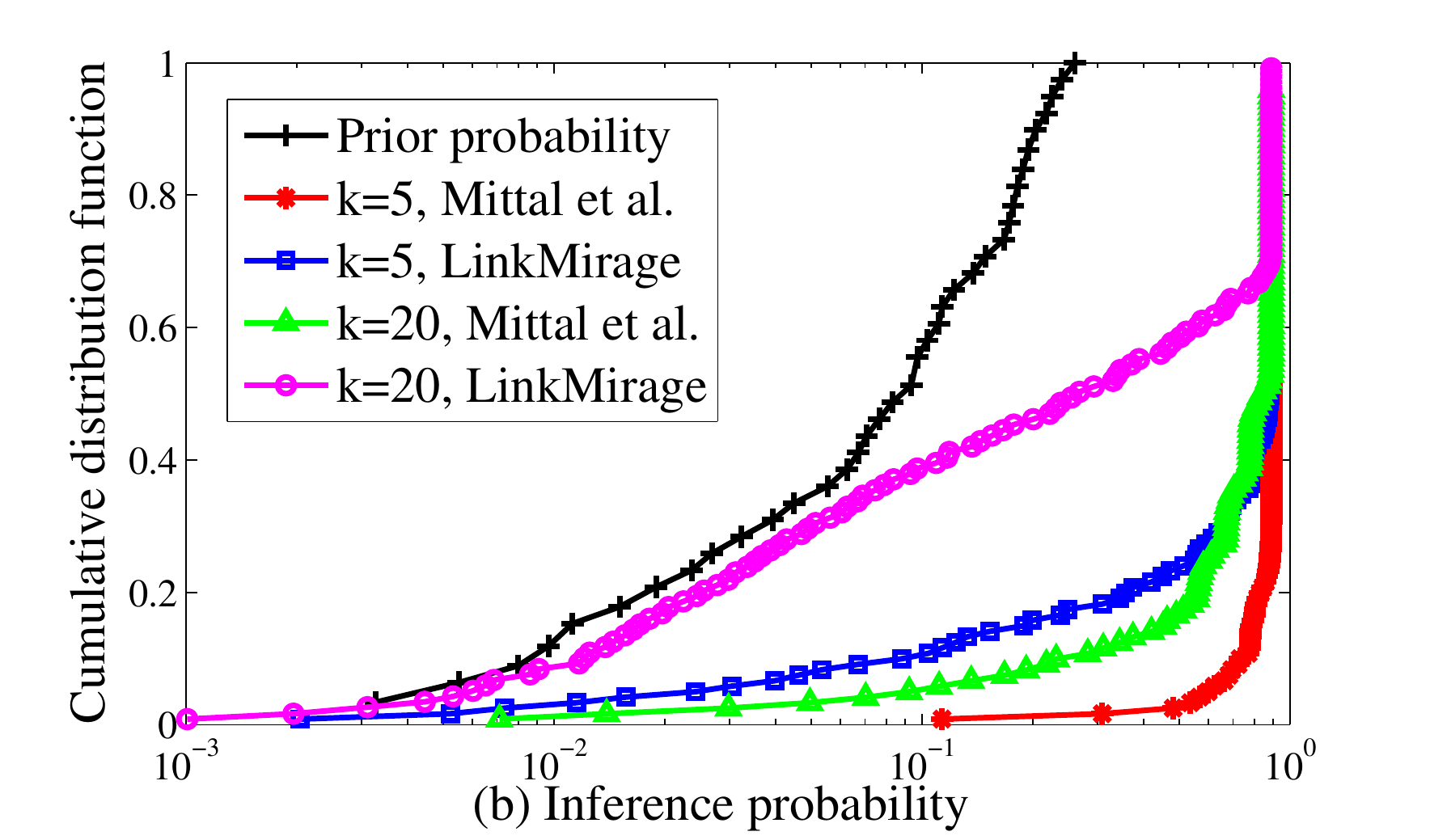}}
\vspace{-0.5em}
\caption{(a),(b) represent the link probability distributions for the
whole Facebook interaction dataset and the sampled
Facebook interaction dataset with 80\% overlap.
We can see that the posterior probability of \system{} is
more similar to the prior probability than the baseline approach. 
}
\label{probability} 
\end{figure*}
\subsection{Anti-Inference Privacy}\label{antiinfer}
First, we consider adversaries that aim to infer link information by leveraging Bayesian inference.
We define the privacy of a link $L_t$ (or a subgraph) in
the $t$-th graph instance, \new{as the difference between the posterior probability and the prior probability of the existence
of the link (or a subgraph), computed by the adversary using
its prior information $W$, and the knowledge of the perturbed
graph sequence $\{{G}^\prime_i\}_{i=0}^t$.} Utilizing Bayesian
inference, we have
\begin{mydef}\label{bayesianprivacy}
\new{For link $L_t$ in the original graph sequence $G_0, \cdots, G_t$ and the adversary's prior information $W$, 
the anti-inference privacy $\mathrm{Privacy}_{{ai}}$ for the perturbed graph sequence $G_0^\prime,\cdots,
G_t^{\prime}$ is evaluated
by the similarity between the posterior probability
$P(L_t|\{{G}^\prime_i\}_{i=0}^t,W)$ and the prior probability $P(L_t|W)$,
where the posterior probability is}
\begin{equation}
P(L_t|\{{G}^\prime_i\}_{i=0}^t,W) = \frac{P(\{{G}^\prime_i\}_{i=0}^t|L_t,W)\times P(L_t|W)}{P(\{{G}^\prime_i\}_{i=0}^t|W)}
\end{equation}
Higher similarity implies better anti-inference privacy. 
\end{mydef}
\indent The difference between the posterior probability and the
prior probability represents the information leaked by the perturbation mechanism. \new{Similar intuition has been mentioned in~\cite{li2013membership}.}
Therefore, the posterior probability should not differ much from the
prior probability. \\
\indent In the above expression, $P(L_t|W)$ is the prior probability of
the link, which can be computed based on the known structural properties
of social networks, for example, by using link prediction algorithms \cite{liben:JASIST07}.
Note that $P(\{{G}^\prime_i\}_{i=0}^t|W)$ is a normalization constant that can
be analyzed by sampling techniques. The key challenge is to
compute $P(\{{G}^\prime_i\}_{i=0}^t|L_t,W)$\footnote{\final{The detailed process for computing the posterior probability can be found in~\cite{Mittal:NDSS13}}}.\\
\indent For evaluation, we consider a special case where the adversary's prior is
the entire time series of original graphs except the
link $L_t$ (which is the link we want to quantify privacy for, and $L_t=1$ denotes
the existence of this link while $L_t=0$ denotes the non-existence of this link). {Such prior information can be extracted from personal public information, Facebook related information or other application-related information as stated in~\cite{burattin:arxiv14}.}
Note that this is a very strong adversarial prior, which would lead
to the worst-case analysis of link privacy. Denoting
$\{\widetilde{G}_{i}(L_t)\}_{i=0}^t$ as the prior which
contains all the information except $L_t$, we have the posterior probability of link $L_t$ under the worst case is
\begin{figure}[!t]
\renewcommand{\captionfont}{\footnotesize}
\centering
\includegraphics[width=2.4in,height=1.4in]{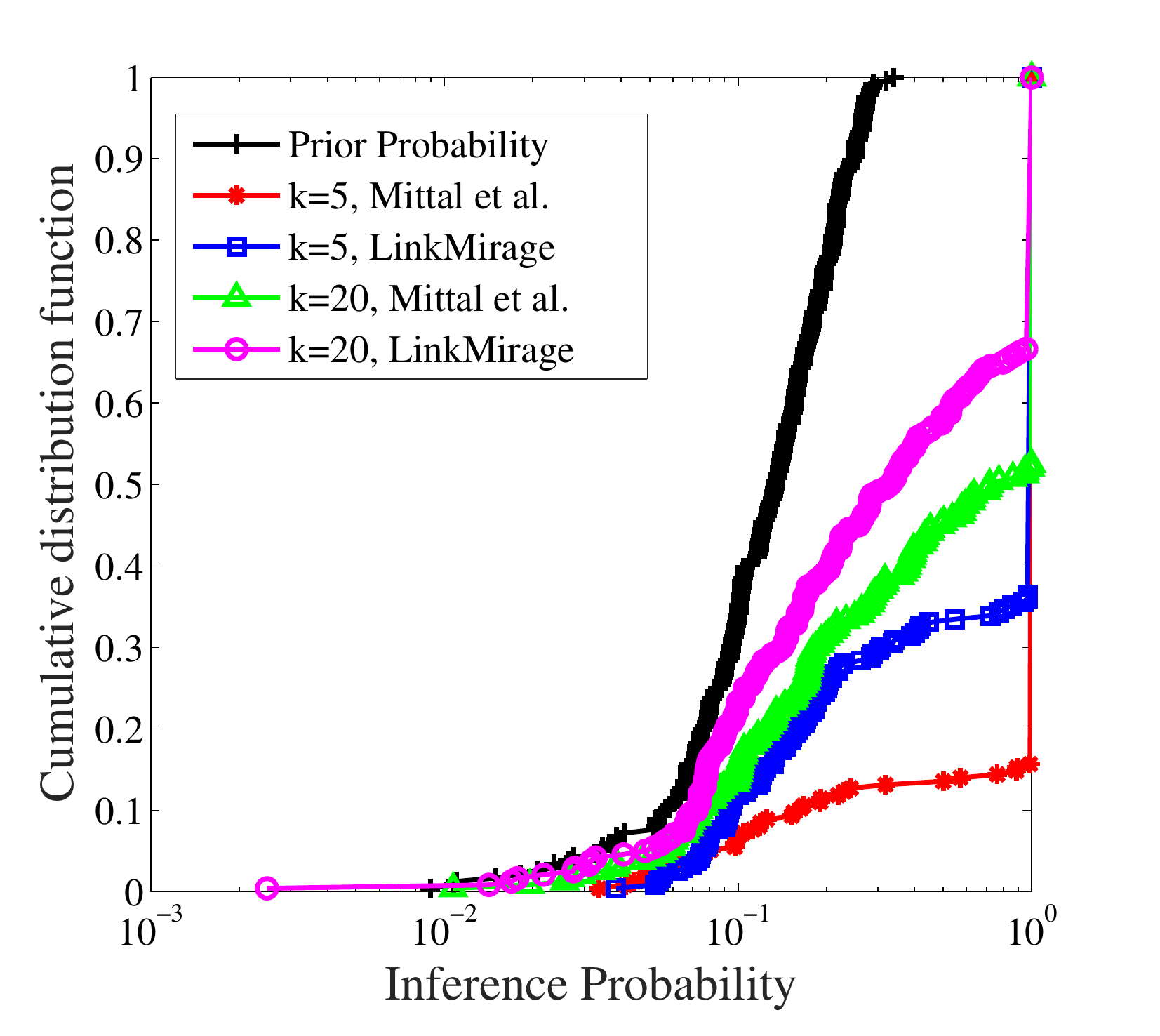}
\DeclareGraphicsExtensions.
\vspace{-0.5em}
\caption{Link probability distribution for the Google+ dataset under the adversary's prior information extracted from the social-attribute network model in \cite{gong:IMC12}.}
\label{context}
\end{figure}
\begin{displaymath}
\begin{aligned}
&P(L_t|\{{G}^\prime_i\}_{i=0}^t,\{\widetilde{G}_i(L_t)\}_{i=0}^t)\\
&=\frac{P(\{{G}^\prime_i\}_{i=0}^t|, L_t,\{\widetilde{G}_i(L_t)\}_{i=0}^t)\times P(L_t|\{\widetilde{G}_i(L_t)\}_{i=0}^t)}{P(\{G^\prime_i\}_{i=0}^t|\{\widetilde{G}_i(L_t)\}_{i=0}^t)}
\end{aligned}
\end{displaymath}
where
\begin{displaymath}
\begin{aligned}
&P(\{{G}^\prime_i\}_{i=0}^t|L_t,\{\widetilde{G}_i(L_t)\}_{i=0}^t)=P(G^\prime_0|\widetilde{G}_0(L_t))\times\\
&P(G^\prime_1|G^\prime_0,\widetilde{G}_0(L_t), \widetilde{G}_1(L_t))\cdots P(G^{\prime}_t|G^\prime_{t-1},\widetilde{G}_{t-1}(L_t), \widetilde{G}_t(L_t))
\end{aligned}
\end{displaymath}
Therefore, the objective of perturbation algorithms is to
make $P(L_t|\{G^\prime_i\}_{i=0}^t,\{\widetilde{G}_i(L_t)\}_{i=0}^t)$
close to $P(L_t|\{\widetilde{G}_i(L_t)\}_{i=0}^t)$.\\
{\bf{Comparison with previous work:}} Fig.~\ref{probability} shows the posterior probability distribution for
the whole Facebook graph sequence and the sampled Facebook graph sequence
with 80\% overlapping ratio, respectively. We computed the prior probability using the link prediction method
in \cite{liben:JASIST07}. We can see that the
posterior probability corresponding to
\system{} is closer to the prior probability than that of the method of Mittal et al.~\cite{Mittal:NDSS13}.
In Fig.~\ref{probability}(b), taking the point where the link
probability equals $0.1$, the distance between the posterior CDF and
the prior CDF for the static approach is a factor of $3$
larger than \system{} ($k=20$).
 Larger perturbation degree $k$ improves privacy
and leads to smaller difference with the prior probability. Finally, by
comparing Fig.~\ref{probability}(a) and (b),
we can see that larger overlap in the graph sequence improves the privacy
benefits of \system{}.
\begin{figure*}[!t]
\renewcommand{\captionfont}{\footnotesize}
\centering{
\label{tempoequiwhole} 
\includegraphics[width=2.7in,height=1.3in]{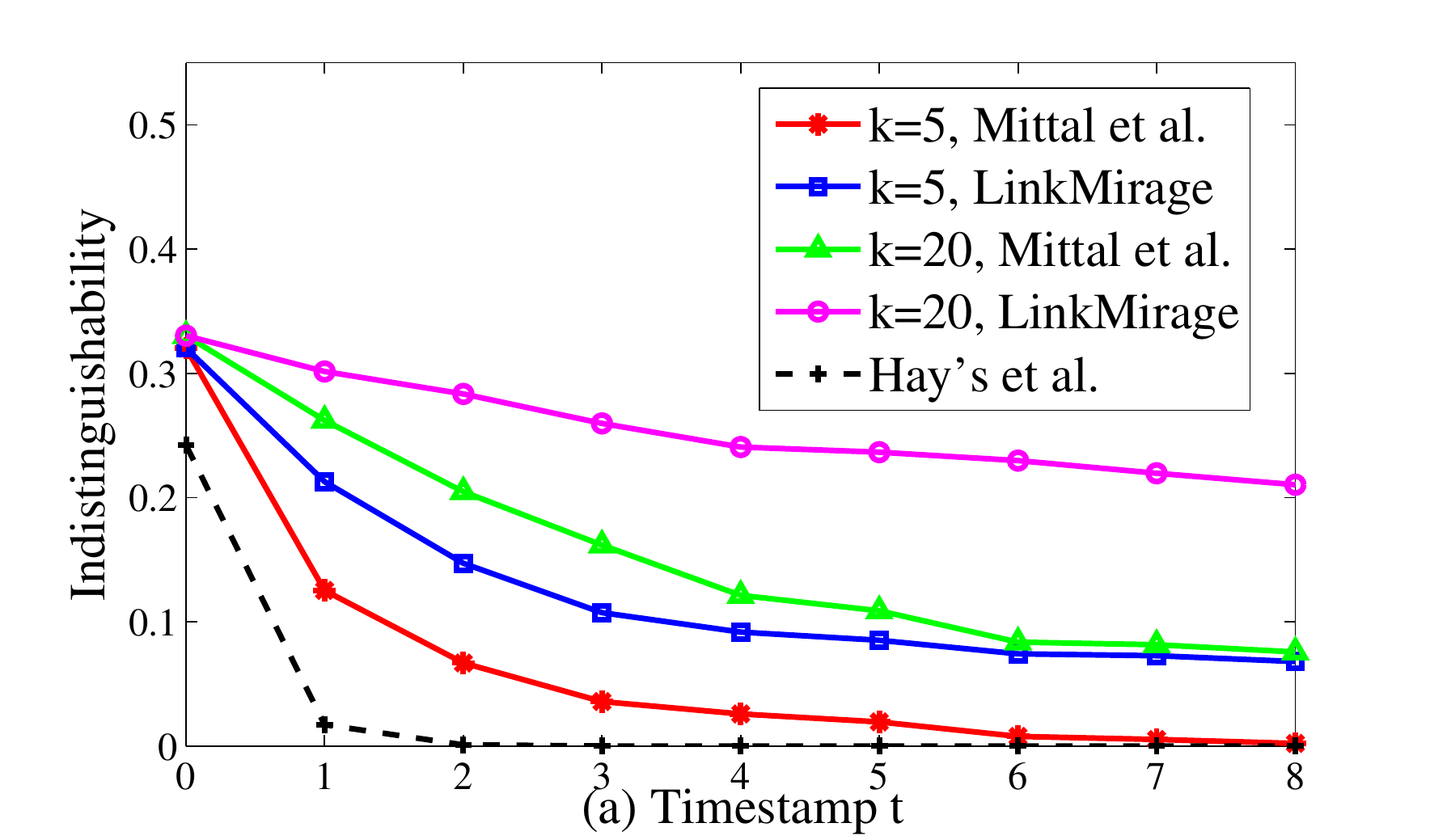}}
\hspace{0.5in}{
\label{tempoequisub} 
\includegraphics[width=2.7in,height=1.3in]{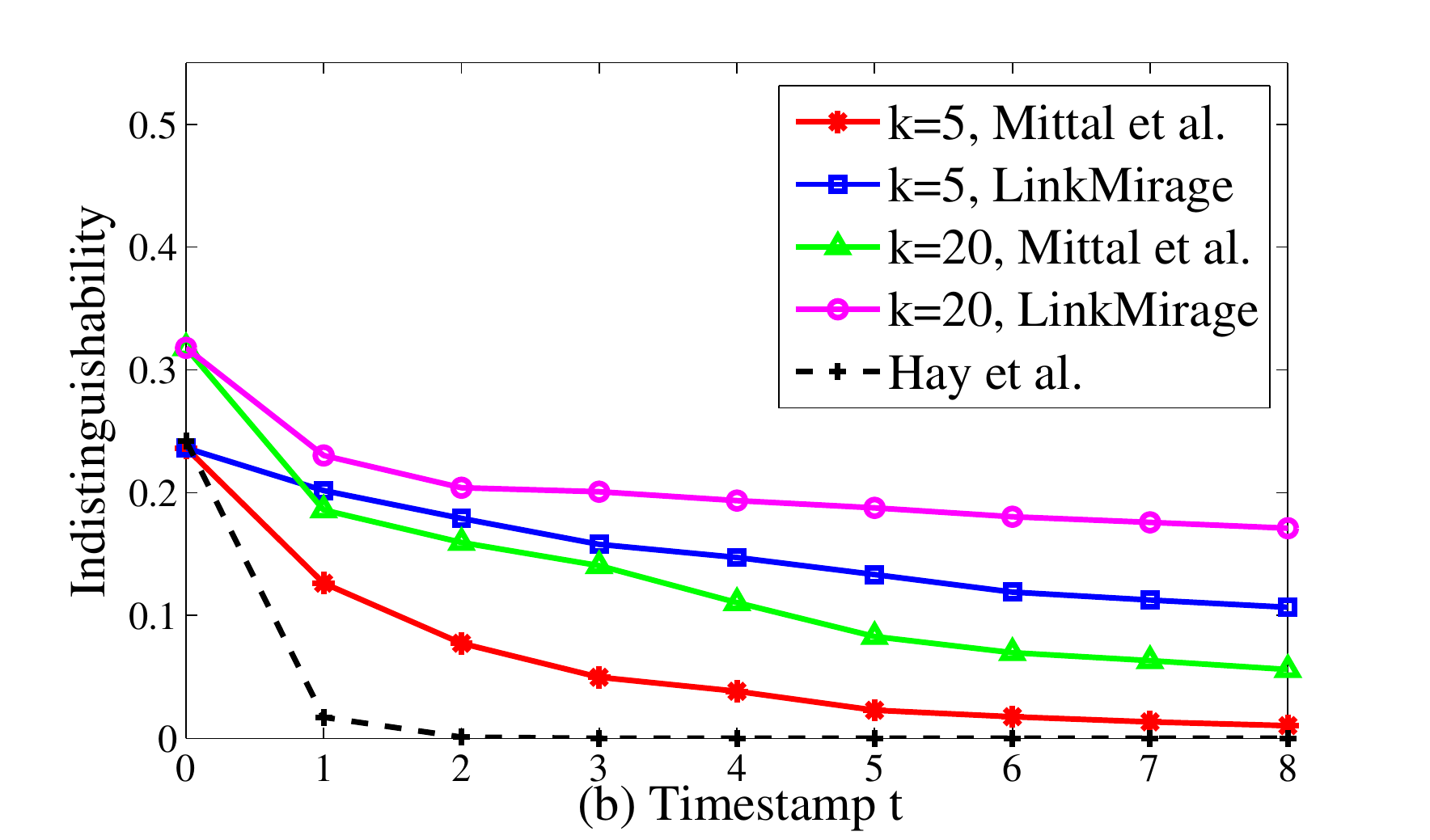}}
\vspace{-0.5em}
\caption{(a),(b) represent the temporal indistinguishability for the whole
Facebook interaction dataset and the sampled Facebook interaction
 dataset with 80\% overlap. 
Over time, the adversary has more information, resulting in decreased indistinguishability.
We can also see that \system{} has higher indistinguishability than the static method
and the Hay's method in \cite{Hay:csfacultypublication07},
although it still suffers from some information leakage.}
\label{tempo_equivocation} 
\end{figure*}
\indent {We also compare with the work of Hay et al. in \cite{Hay:csfacultypublication07},
which randomizes the graph with $r$ real
links deleted and another $r$ fake links introduced.
\new{The probability for a real link to be preserved in the
perturbed graph is $1-{r}/{m}$, which should not
be small otherwise the utility would not be preserved.
Even considering ${r}/{m}=0.5$ (which would substantially hurt utility~\cite{Hay:csfacultypublication07}), the posterior probability
for a link using the method of Hay et al. would be $0.5$, \emph{even without prior information}. 
In contrast, {our analysis for \system{} considers a worst-case prior, and shows that the
 posterior probability is smaller than $0.5$ for more than 50\% of the links when $k=20$ in Fig.~\ref{probability}}. Therefore, our \system{} provides significantly higher privacy than the work of Hay et al.}}\\
\new{{\bf{Adversaries with structural and contextual information:}} Note that our analysis so far focuses on quantifying link-privacy under an adversary with prior information about the original network structure (including link prediction capabilities). In addition, some adversaries may also have access to contextual information about users in the social network, such as user attributes, which can also be used to predict network links (e.g., social-attribute network prediction model in \cite{gong:IMC12}). We further computed the prior probability using such social-attribute network prediction model in \cite{gong:IMC12} and showed the link probability for the Google+ dataset in Fig.~\ref{context}. The posterior probability of our \system{} is closer to the prior probability and thus \system{} achieves better privacy performance than previous work.}   
\subsection{Indistinguishability}\label{indis}
Based on the posterior probability of a link under the worst case $P(L_t|\{G^\prime_i\}_{i=0}^t,\{\widetilde{G}_i(L_t)\}_{i=0}^t)$, we need to qualify the privacy metric for adversaries who aim to distinguish the posterior probability with the prior probability.
Since our goal is to reduce the information leakage
of $L_t$ based on the perturbed graphs $\{G^{\prime}_i\}_{i=0}^t$
and the prior knowledge $\{\widetilde{G}_i(L_t)\}_{i=0}^t$, we
consider the metric of \textit{indistinguishability} to quantify
privacy, which can be evaluated by
the conditional entropy of a private message given the
observed variables \cite{Cover:IT12}. The objective for an obfuscation scheme is to maximize the indistinguishability
of the unknown input ${I}$ given the observables ${O}$,
i.e. $H({I}|{O})$ (where $H$ denotes entropy of a variable \cite{Cover:IT12}).
{Here, we define our metric for link privacy as
\begin{mydef}\label{equiprivacy}
\new{The \emph{indistinguishability} for a link $L_t$ in the original graph $G_t$ that the adversary can infer from the perturbed graph $G_t^{\prime}$ under the adversary's prior information
$\{\widetilde{G}_i(L_t)\}_{i=0}^t$ is defined as
$\mathrm{Privacy}_{\mathrm{id}}=H(L_t|\{G^\prime_i\}_{i=0}^t,\{\widetilde{G}_i(L_t)\}_{i=0}^t)$.}
\end{mydef}
\indent Furthermore, we quantify the behavior of indistinguishability
over time. For our analysis, we continue to consider the worst case prior
of the adversary knowing the entire graph sequence except the link $L_t$.
To make the analysis tractable, we add another condition that if the
link $L$ exists, then it exists in all the graphs (link deletions are
rare in real world social networks). For a large-scale graph,
only one link would not affect the clustering result. Then, we have
\begin{mythe}
The \emph{indistinguishability} decreases with time,
\begin{equation}\label{tempo_equi}
H(L|\{G_i^{\prime}\}_{i=0}^t, \{\widetilde{G}_i(L)\}_{i=0}^t)
\ge H(L|\{G_i^{\prime}\}_{i=0}^{t+1}, \{\widetilde{G}_i(L)\}_{i=0}^{t+1})
\end{equation}
\end{mythe}
\indent The inequality follows from the theorem
\textit{conditioning reduces entropy} in \cite{Cover:IT12}.
Eq.\ref{tempo_equi} shows that the \emph{indistinguishability} would
not increase as time evolves.
The reason is that over time, multiple perturbed graphs
can be used by the adversary to infer more information about link $L$.\\
\indent Next, we theoretically show why \system{}
 has better privacy performance than the static method.
 For each graph $G_t$, denote the
 perturbed graphs using \system{} and the static method as
$G^{\prime}_t, G^{\prime,s}_{t}$, respectively.
\begin{mythe}
The \emph{indistinguishability} for \system{}
is greater than that for the static perturbation method, i.e.
\begin{equation}\label{equivocation}
H(L_t|\{{G}^{\prime}_{i}\}_{i=0}^t, \{\widetilde{G}_i(L_t)\}_{i=0}^t)
\ge H(L_t|\{G^{\prime,s}_{i}\}_{i=0}^t, \{\widetilde{G}_i(L_t)\}_{i=0}^t)
\end{equation}
\end{mythe}
\begin{proof}
In \system{}, the perturbation for the current graph
$G_t$ is based on perturbation for $G_{t-1}$. Let us denote the \textit{changed} subgraph between
$G_{t-1}, G_t$ as $G_{t\mathrm{-ch}}$, then
\begin{displaymath}
\begin{aligned}
&H(L_t|\{{G}^{\prime}_{i}\}_{i=0}^t, \{\widetilde{G}_i(L_t)\}_{i=0}^t)\\
=&H(L_t|\{G^\prime_{i}\}_{i=0}^{t-2}, G^{\prime}_{t-1},G^{\prime}_t-G^{\prime}_{t\mathrm{-ch}}, G^{\prime,s}_{t\mathrm{-ch}}, \{\widetilde{G}_i(L_t)\}_{i=0}^t)\\
=&H(L_t|\{G^\prime_{i}\}_{i=0}^{t-1}, G_{t\mathrm{-ch}}^{\prime,s}, \{\widetilde{G}_i(L_t)\}_{i=0}^t)\\
\ge &H(L_t|\{G^\prime_{i}\}_{i=0}^{t-1},{G_t}^{\prime,s}, \{\widetilde{G}_i(L_t)\}_{i=0}^t)\\
\ge  &H(L_t|\{G^{\prime,s}_{i}\}_{i=0}^t, \{\widetilde{G}_i(L_t)\}_{i=0}^t)
\end{aligned}
\end{displaymath}
where the first inequality also comes from the theorem
\textit{conditioning reduces entropy} in \cite{Cover:IT12}. The second inequality
generalizes the first inequality from a snapshot $t$ to the entire sequence.
From Eq.\ref{equivocation}, we can see that \system{}
may offer superior \textit{indistinguishability}
compared to the static perturbation, and thus provides
higher privacy.
\end{proof}
{\bf{Comparison with previous work:}} Next, we experimentally analyze our \emph{indistinguishability} metric
over time. Fig.~\ref{tempo_equivocation} depicts the \emph{indistinguishability} metric using
the whole Facebook graph sequence and the sampled Facebook graph sequence
with 80\% overlap. We can see that the static perturbation leaks more
information over time. In contrast, the selective perturbation achieves
significantly higher \emph{indistinguishability}. In Fig.~\ref{tempo_equivocation}(a),
after 9 snapshots, and using $k=5$, the \emph{indistinguishability} of the
static perturbation method is roughly $1/{10}$ of
the indistinguishability of \system{}.
This is because selective perturbation explicitly takes the temporal
evolution into consideration, and stems privacy degradation via
the selective perturbation step. Comparing Fig.~\ref{tempo_equivocation}(a)
and (b), \system{}
has more advantages for larger overlapped graph sequence.\\
\indent \new{We also compare with the work of Hay et al. in \cite{Hay:csfacultypublication07},
For the first timestamp, the probability for a real link
to be preserved in the anonymized graph is $1-{r}/{m}$.
As time evolves, the probability would decrease to
$(1-{r}/{m})^t$. Combined with the prior probability, the corresponding \emph{indistinguishability} for the method of Hay et al. is
shown as the black dotted line in Fig.~\ref{tempo_equivocation},
which converges to 0 very quickly
(we also consider ${r}/{m}=0.5$ which would substantially hurt utility~\cite{Hay:csfacultypublication07})
Compared with the work of Hay et al,
\system{}
significantly improves privacy performance. {Even when $t=1$, \system{} with $k=20$ achieves up to 10x improvement over the approach of Hay et al. in the \emph{indistinguishability} performance}}.
\begin{figure*}[!t]
\renewcommand{\captionfont}{\footnotesize}
\centering{
\label{SPgoogle} 
\includegraphics[width=2.7in,height=1.3in]{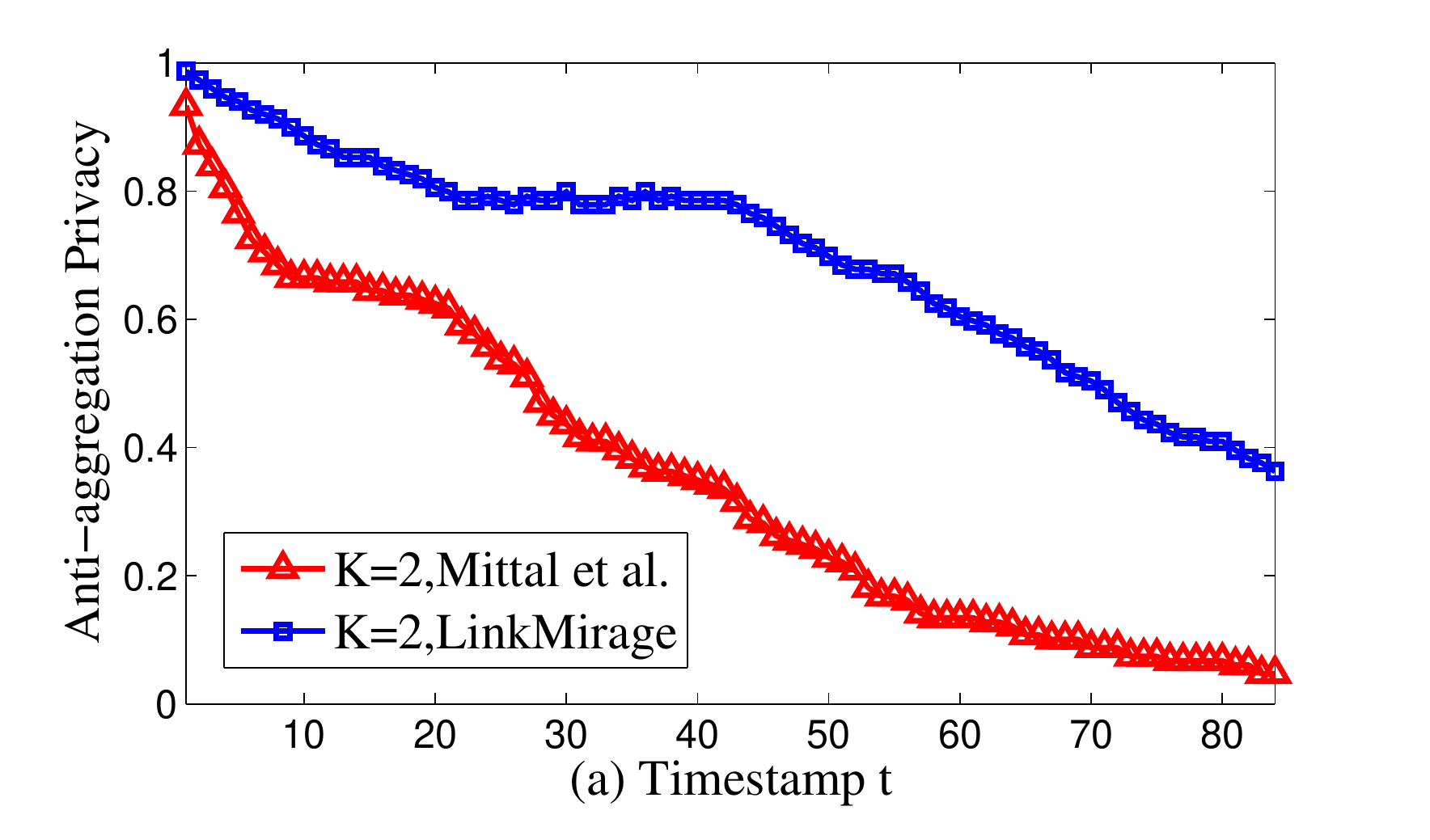}}
\hspace{0.5in}
\centering{
\label{SPfb} 
\includegraphics[width=2.7in,height=1.3in]{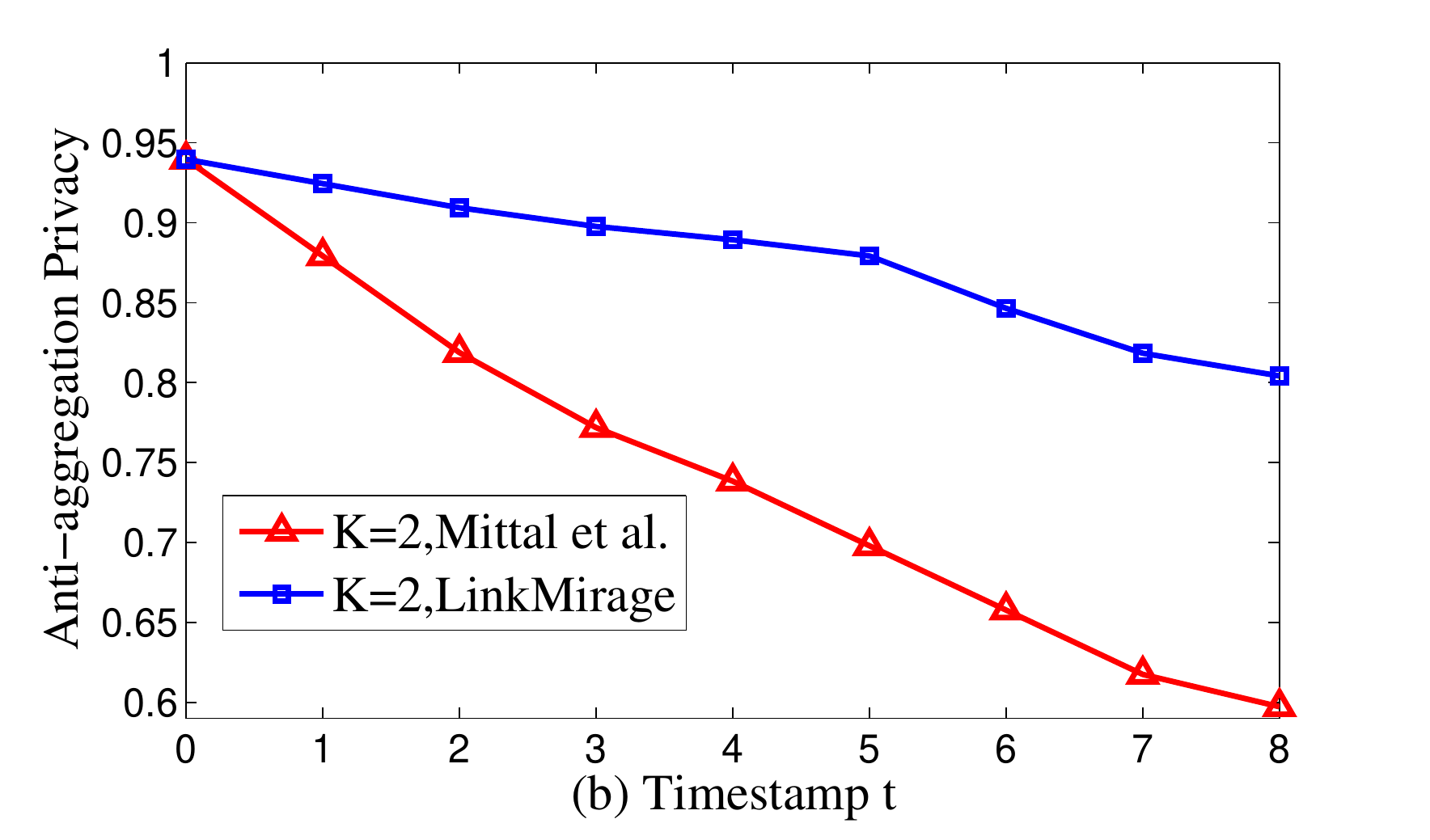}}
\vspace{-0.8em}
\caption{(a)(b) show the temporal anti-aggregation privacy for
the Google+ dataset and the Facebook dataset, respectively.
The anti-aggregation privacy decreases as time evolves
because more information is leaked with more perturbed
graphs available. Leveraging selective perturbation, \system{} achieves much better anti-aggregation
privacy than the static baseline method.}
\label{SP}
\end{figure*}
\subsection{Anti-aggregation Privacy}\label{antiaggre}
{Next, we consider the adversaries who try to aggregate all the previously published graphs to infer more information.  Recall that after community detection in our algorithm,
we anonymize the links by leveraging the $k$-hop random walk.
Therefore, the perturbed graph $G^\prime$ is actually a sampling of
the $k$-hop graph $G^k$, where the $k$-hop graph $G^k$ represents
graph where all the $k$-hop neighbors in the original graph
are connected. It is intuitive that a larger difference between
$G^k$ and $G^\prime$ represents better privacy. Here, we utilize
the distance between the corresponding transition probability matrices
$\|P^k_t-P_t^\prime\|_{\mathrm{TV}}$ \footnote{\final{We choose the total variance distance to evaluate the statistical distance between  $P^k_t$ and $P_t^\prime$ as in \cite{Mittal:NDSS13}.}} to measure this difference.
And we extend the definition of total variance \cite{horn:matrix12} from vector to matrix by
averaging total variance distance of each row in the matrix, i.e. $\|P^k_t-P_t^\prime\|_{\mathrm{TV}}=\frac{1}{|V_t|}{\sum_{v=1}^{|V_t|}\|P_t^k(v)-P_t^\prime(v)\|_{\mathrm{TV}}}$,
where $P_t^k(v), P_t^\prime(v)$ denotes the $v$-th row of $P_t^k, P_t^\prime$.
We then formally define the anti-aggregation privacy as
\begin{mydef}
The anti-aggregation privacy for a perturbed graph $G_t^{\prime}$
with respect to the original graph $G_t$ and the perturbation parameter $k$
is 
$\mathrm{{Privacy}_{aa}}(G_t, G_t^\prime, k)={\|P_t^k-P_t^\prime\|_{\mathrm{TV}}}.$
\end{mydef}
}
The adversary's final objective is to obtain an estimated measurement of the original
graph, e.g. the estimated transition probability matrix $\hat P_t$
which satisfies ${\hat P_t}^k=P_t^\prime$. A straightforward
manner to evaluate privacy is to compute the estimation error of the transition probability matrix
i.e. $\|P_t-{\hat P_t}\|_{\mathrm{TV}}$. We can derive the relationship between the anti-aggregation
privacy and the estimation error as ({we defer the proofs to the Appendix to improve readability.})
\begin{mythe}
The anti-aggregation privacy is a lower bound of the estimation error for the adversaries, and
\begin{equation}
{\|P_t^k-P_t^\prime\|_{\mathrm{TV}}}\le k\|P_t-\hat P_t\|_{\mathrm{TV}}
\end{equation}
\end{mythe}
We further consider the network evolution where the adversary can combine all the perviously
perturbed graphs together to extract more $k$-hop information
of the current graph. Under this situation, a strategic
methodology for the adversary is to
combine the perturbed graph series $G_0^\prime,\cdots,G_t^\prime$,
to construct a new perturbed graph $\breve G_t^\prime$, where
$\breve G_t^\prime=\bigcup_{i=0,1\cdots,t} G_i^\prime$. The combined
perturbed graph $\breve G_t^\prime$ contains more
information about the $k$-hop graph $G^k_t$ than $G_t^\prime$. Correspondingly,
the transition probability matrix $\breve P_t^\prime$ of the combined perturbed graph $\breve G_t^\prime$ 
would provide more information than
$P_t^\prime$. That is to say, the \emph{anti-aggregation privacy} decreases with time.\\
{\bf{Comparison with previous work:}} We evaluate the anti-aggregation privacy of \system{} on both the Google+ dataset and the Facebook
dataset. Here we perform our experiments based on a conservative assumption that a link
always exists after it is introduced. The anti-aggregation privacy decreases with time
 since more information about the $k$-hop neighbors of
 the graph is leaked as shown in Fig.~\ref{SP}. Our selective
 perturbation preserves correlation
 between consecutive graphs, therefore leaks less information
 and achieves better privacy than the static baseline method.
 For the Google+ dataset, the anti-aggregation privacy for the method of
 Mittal et al. is only $1/10$ of \system{} after 84 timestamps.
\subsection{Relationship with Differential Privacy}
Our anti-inference privacy analysis considers the worst-case
adversarial prior to infer the existence of a
link in the graph. Next, we uncover a novel relationship
between this anti-inference privacy 
 and differential
privacy.\\
\indent Differential privacy is a popular theory to
evaluate the privacy of a perturbation scheme \cite{dwork:Springer06, dwork:ACM09,liu2016dependence,}.
The framework of differential privacy defines \emph{local sensitivity}
 of a query function $f$ on a dataset $D_1$ as the maximal
$|f(D_1)-f(D_2)|_1$ for all $D_2$ differing
from $D_1$ in at most one element
$df=\max_{D_2}\|f(D_1)-f(D_2\|_1$.
Based on the theory of differential privacy, a mechanism that adds independent Laplacian
noise with parameter $df/\epsilon$ to the query function $f$, satisfies
$\epsilon$-differential privacy. The degree of added noise, which determines
the utility of the mechanism, depends on the local sensitivity.
To achieve a good utility as well as privacy, the local sensitivity $df$
should be as small as possible. \final{The following lemma demonstrates the effectiveness of worst-case Bayesian analysis since the objective for good utility-privacy balance under our worst-case Bayesian analysis is equivalent to under differential privacy.}
\begin{myrem}
The requirement for good utility-privacy
balance in differential privacy is equivalent to the objective
of our Bayesian analysis under the worst case. ({We defer the proofs to Appendix to improve readability.})
\end{myrem}
}
\subsection{Summary for Privacy Analysis}
\begin{enumerate}[$\bullet$]
\item{\system{} provides rigorous privacy guarantees to defend against adversaries who have prior information about the original graphs, and the adversaries who aim to combine multiple released graphs to infer more information.}
\item{\system{} shows significant privacy advantages in \emph{anti-inference privacy, indistinguishability and anti-aggregation privacy}, by outperforming previous methods by a factor up to $10$.}
\end{enumerate}

\vspace{-1em}
\section{Applications}\label{app}
{Applications such as anonymous communication \cite{Dingledine:USENIX04,nagaraja:pet07,Mittal:piscesNDSS13} and vertex anonymity mechanisms~\cite{zhou:SIGKDD08,sala:imc11,liu:sigmod08} can utilize \system{} to obtain 
the entire obfuscated social graphs. Alternatively, each individual user can query \system{} for his/her perturbed neighborhoods to set up distributed social relationship based applications such as SybilLimit \cite{Yu:IEEES&P08}. Further, the OSN providers can also leverage \system{} to perturb the original social topologies only once and support multiple privacy-preserving graph analytics, e.g., privately compute the pagerank/modularity of social networks. }
\begin{figure}[!t]
\renewcommand{\captionfont}{\footnotesize}
\centering
\includegraphics[width=2.7in,height=1.3in]{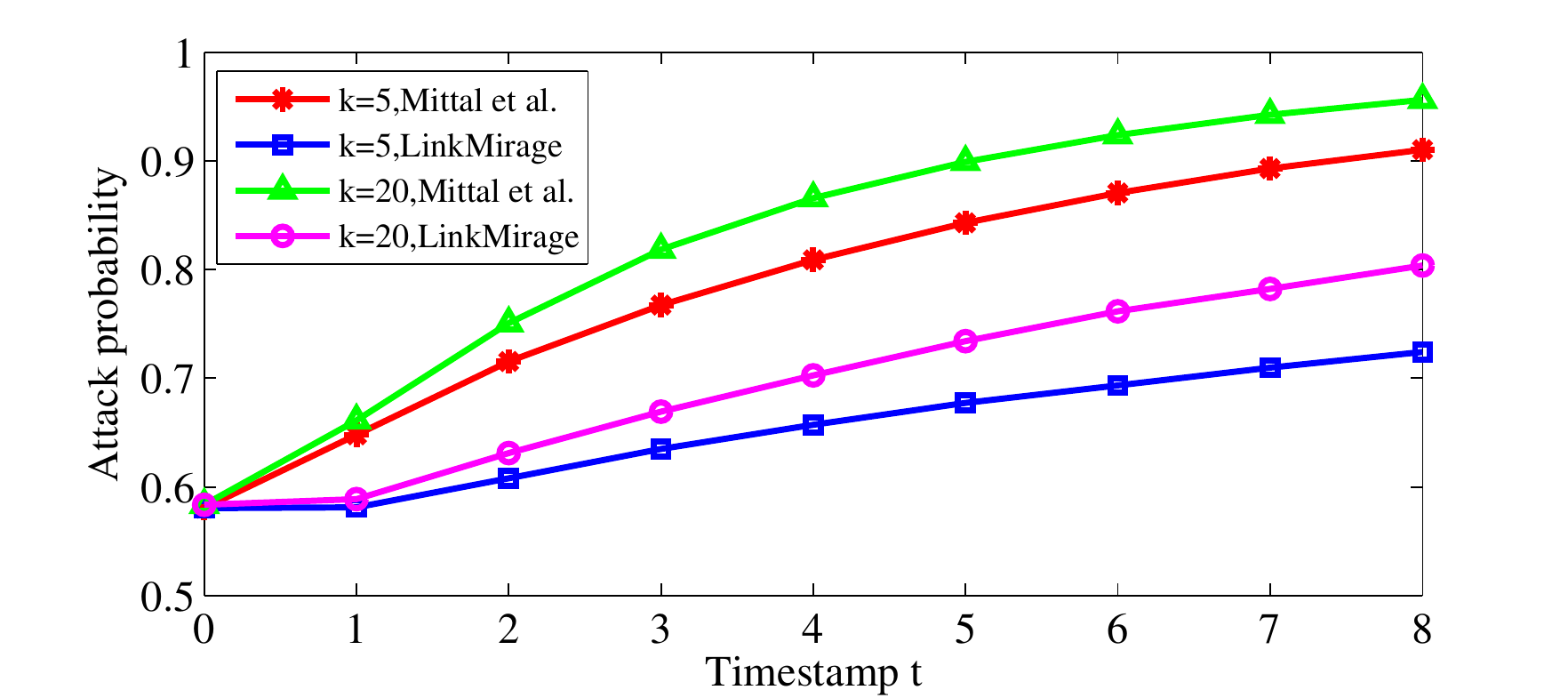}
\DeclareGraphicsExtensions.
\caption{The worst case probability of deanonymizing users'
communications ($f=0.1$). Over time, \system{} provides
better anonymity compared to the static approaches.}
\label{pisces}
\end{figure}
\begin{figure*}[!t]
\renewcommand{\captionfont}{\footnotesize}
\centering{
\label{fig10a} 
\includegraphics[width=2.7in,height=1.3in]{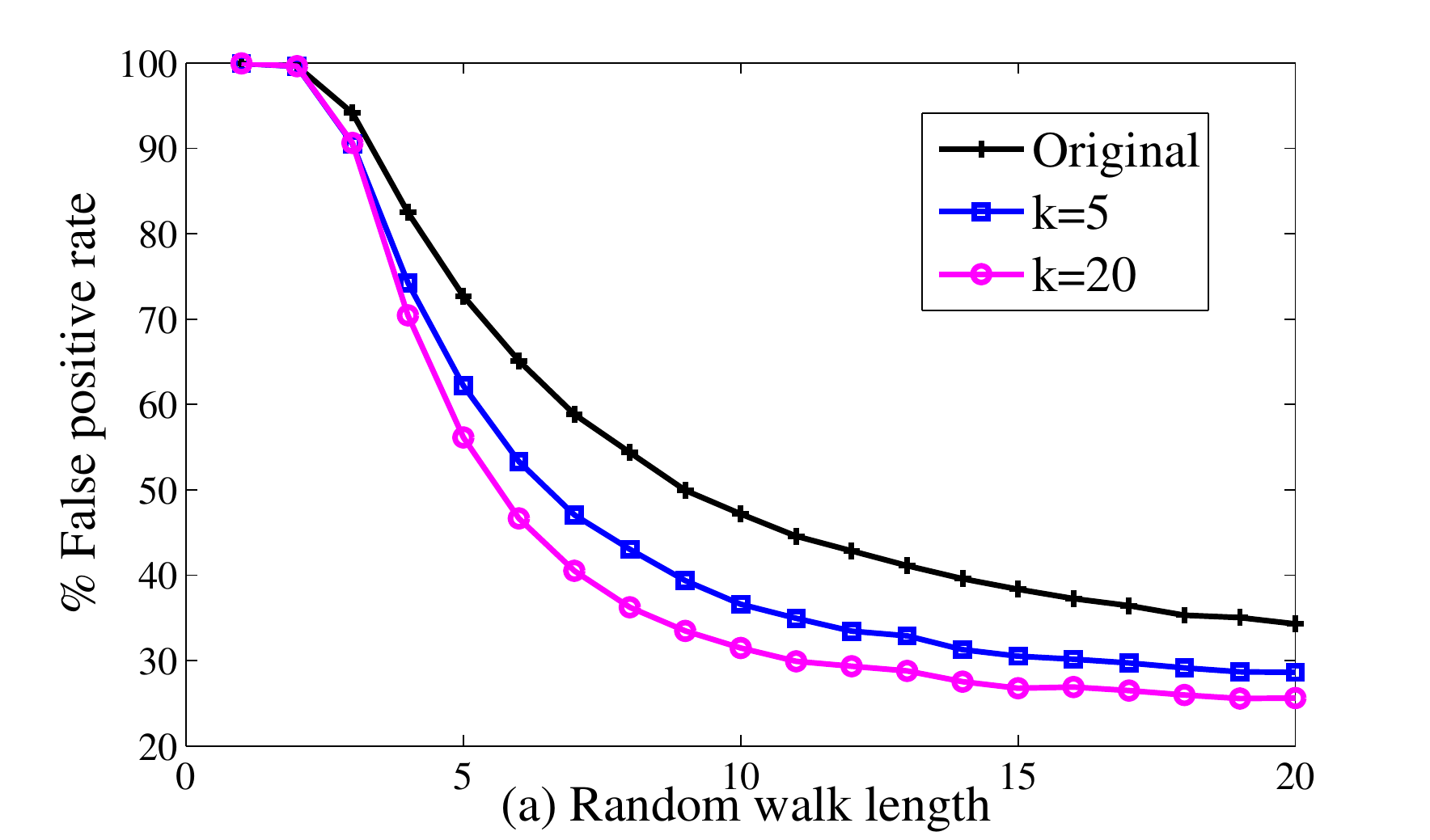}}
\hspace{0.5in}{
\label{fig10b} 
\includegraphics[width=2.7in,height=1.3in]{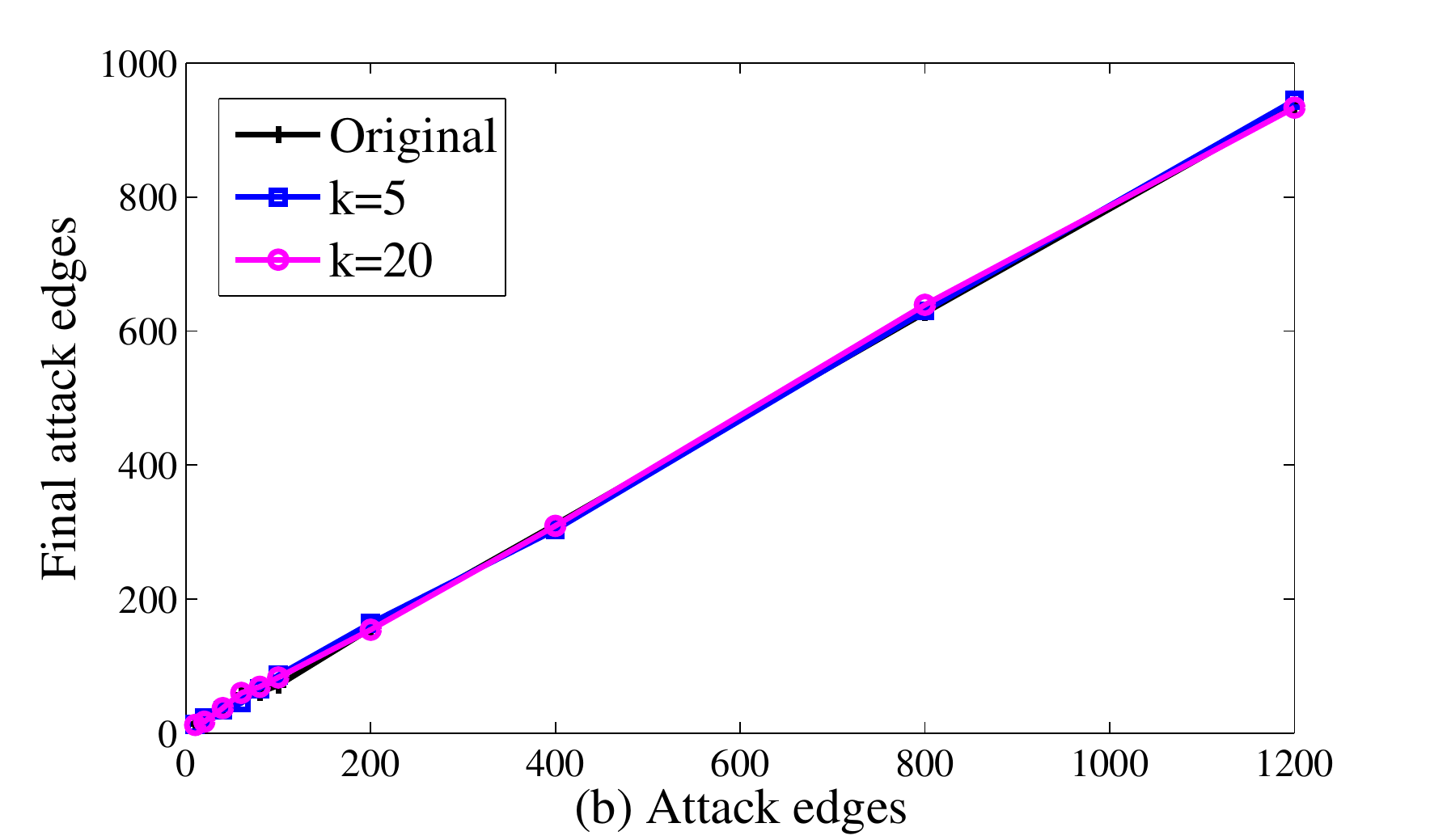}}
\vspace{-0.5em}
\caption{(a) shows the false positive rate
for Sybil defenses. We can see that the perturbed graphs
have lower false positive rate than the original graph.
Random walk length is proportional
to the number of Sybil identities that can be
inserted in the system. (b) shows that the final attack edges
are roughly the same for the perturbed graphs
and the original graphs.}
\label{sybildefense} 
\end{figure*}
\subsection{Anonymous Communication \cite{Dingledine:USENIX04,nagaraja:pet07,Mittal:piscesNDSS13}}
 As a concrete
application, we consider the problem of anonymous
communication~\cite{Dingledine:USENIX04,nagaraja:pet07,Mittal:piscesNDSS13}. Systems
for anonymous communication aim to improve user's privacy by hiding the communication
 link between the user and the remote destination.
Nagaraja et al. and others~\cite{Dingledine:USENIX04,nagaraja:pet07,Mittal:piscesNDSS13} have suggested that the
security of anonymity systems can be improved by leveraging
users' trusted social contacts. \\
\indent We envision that our work can be a key enabler for the
design of such social network based systems,
while preserving the privacy of users' social relationships.
We restrict our analysis to
low-latency anonymity systems that leverage social links,
such as the Pisces protocol~\cite{Mittal:piscesNDSS13}.\\
\indent Similar to the Tor protocol, users in Pisces rely
on proxy servers and onion routing for anonymous
communication. However, the relays involved in the
onion routing path are chosen by performing a random
walk on a trusted social network topology.
Recall that \system{} better preserves the evolution of temporal graphs in Fig.~\ref{visualutility}. 
We now show that this translates into improved
anonymity over time, by performing an analysis of the degradation
of user anonymity over multiple graph snapshots. For each graph
snapshot, we consider a worst case anonymity analysis as follows:
if a user's neighbor in the social topology is malicious, then
over multiple communication rounds (within that graph instance)
its anonymity will be compromised using state-of-the-art traffic
analysis attacks~\cite{wright:sp03}.
Now, suppose that all of a user's neighbors in the first graph instance are honest.
As the perturbed graph sequence evolves, there is further
potential for degradation of user anonymity since in the
subsequent instances, there is a chance of the user connecting
to a malicious neighbor. Suppose the probability for
 a node to be malicious is $f$. Denote $n_t(v)$ as the distinct
neighbors of node $v$ at time $t$. For a temporal graph sequence,
the number of the union neighbors $\cup_{k=0}^{t}n_k(v)$
of $v$ increases with time, and the probability for $v$ to be attacked under
the worst case is $P^{attack}_t(v)=1-(1-f)^{|{\cup}_{k=0}^{t}n_k(v)|}$. \new{Note that in practice, the adversary's 
prior information will be significantly less than the worst-case adversary.} \\
\indent Fig.~\ref{pisces} depicts the degradation of
the worst-case anonymity 
with respect to the number of perturbed topologies.
We can see that the attack probability for our method is lower
than the static approach {with a factor up to 2}. This is because over consecutive
graph instances, the users' social neighborhood has
higher similarity as compared to the static approach,
reducing potential for anonymity degradation. Therefore,
\system{} can provide better security for
anonymous communication, and other social trust based applications.
\subsection{Vertex Anonymity \cite{zhou:SIGKDD08,sala:imc11,liu:sigmod08}}\label{vertex}
Previous work for vertex anonymity \cite{zhou:SIGKDD08,sala:imc11,liu:sigmod08} would be defeated by de-anonymization techniques \cite{Narayaran:S&P09,srivatsa:CCS12,nilizadeh:CCS14,ji:CCS14}. \system{} can serve as a formal first step for vertex anonymity, and even improve its defending capability against de-anonymization attacks. 
We apply \system{} to anonymize vertices, i.e.
to publish a perturbed topology without labeling
any vertex. In \cite{ji:CCS14}, Ji et al. modeled the anonymization
 as a sampling process where the
sampling probability $p$ denotes the
probability of an edge in
the original graph $G_o$ to exist in the
anonymized graph $G^\prime$.
\system{} can
also be applied for such model, where
the perturbed graph
$G^\prime$ is sampled from the $k$-hop graph $G^k$ (corresponding
to $G_o$). \\
\indent They also derived a theoretical bound of the sampling probability $p$
for perfect de-anonymization, and found that a weaker bound is needed with a larger value
of the sampling probability $p$. Larger $p$ implies that 
$G^\prime$ is topologically more similar to $G$, making
it easier to enable a perfect de-anonymization.
When considering social network evolution, the sampling probability
$p$ can be estimated as ${|E(G_0^\prime, \cdots, G_t^\prime)|}/{|E(G_0^k, \cdots, G_t^k)|}$,
where $E(G_0^\prime, \cdots, G_t^\prime)$ are
the edges of the perturbed graph sequence, and
$E(G_0^k, \cdots, G_t^k)$ are the edges of the
$k$-hop graph sequence.
Compared with the static baseline approach, \system{}
selectively reuses information from previously perturbed graphs,
thus leading to smaller overall sampling probability $p$,
which makes it harder to perfectly de-anonymize the graph sequence.
For example, the average sampling probability $p$
for the Google+ dataset (with $k=2$)
is $0.431$ and $0.973$ for \system{} and the static method respectively.
For the Facebook temporal dataset (with $k=3$), the average sampling
probability $p$ is $0.00012$ and $0.00181$ for \system{}
and the static method respectively.
Therefore, {\system{} is more resilient against de-anonymization
attacks even when applied to vertex anonymity, with up to 10x improvement}. 
\subsection{Sybil Defenses \cite{Yu:IEEES&P08}}
Next, we consider Sybil defenses systems which leverage the published social topologies to detect fake accounts in the social networks. Here, we analyze how the use of a perturbed graph changes
the Sybil detection performance of SybilLimit \cite{Yu:IEEES&P08}, which is a representative Sybil defense system. Each user can query \system{} for his/her perturbed friends to set up the implementation of SybilLimit.
 Fig.~\ref{sybildefense}(a)
depicts the false positives (honest users misclassified as Sybils)
with respect to the random walk length in the Sybillimit protocol.
Fig.~\ref{sybildefense}(b)
shows the final attack edges with respect to the attack edges
in the original topology.
We can see that the false positive rate is much
lower for the perturbed graphs than for
the original graph, while the number of
the attack edges stay roughly the same for the
original graph and the perturbed graphs.
The number of Sybil identities that an adversary
can insert is given by $S=g^{\prime}\cdot w^{\prime}$
($g^\prime$ is the number of attack edges and $w^\prime$
is the random walk parameter in the protocol).
Since $g^{\prime}$ stays almost invariant and the
random walk parameter $w^{\prime}$
(for any desired false positive rate) is reduced,
\system{}
improves Sybil resilience and provides the privacy of the
social relationships such that Sybil
defense protocols continue to be applicable (similar to
static approaches whose Sybil-resilience performance have
been demonstrated in previous work).
\begin{table}[!t]\small
\newcommand{\tabincell}[2]{\begin{tabular}{@{}#1@{}}#2\end{tabular}}
\renewcommand{\tabcolsep}{1.5pt}
\centering
\caption{{Modularity of Perturbed Graph Topologies}}
\begin{tabular}{c|c|c|c|c|c}
\hline
Google+&\tabincell{c}{Original\\Graph}&\tabincell{c}{\system{}\\ $k=2$}&\tabincell{c}{\system{}\\ $k=5$}&\tabincell{c}{Mittal et al.\\ $k=2$}&\tabincell{c}{Mittal et al.\\ $k=5$}\\
\hline
Modularity&0.605&0.601&0.603&0.591&0.586\\
\hline
Facebook&\tabincell{c}{Original\\Graph}&\tabincell{c}{\system{}\\$k=5$}&\tabincell{c}{\system{}\\$k=20$}&\tabincell{c}{Mittal et al.\\ $k=5$}&\tabincell{c}{Mittal et al.\\$k=20$}\\
\hline
Modularity&0.488&0.479&0.487&0.476&0.415\\
\hline
\end{tabular}
\label{modularity}
\end{table}
\subsection{Privacy-preserving Graph Analytics \cite{newman:PNAS06,page:stanford99}}
Next, we demonstrate that \system{} can also benefit the OSN providers for privacy-preserving graph analytics. Previous work in \cite{dwork:Springer06,dwork:ACM09} have demonstrated that the implementation of graph analytic algorithms would also result in information leakage. To mitigate such privacy degradation, the OSN providers could add perturbations (noises) to the outputs of these graph analytics. However, if the OSN providers aim to implement multiple graph analytics, the process for adding perturbations to each output would be rather complicated. Instead, the OSN providers can first obtain the perturbed graph by leveraging \system{} and then set up these graph analytics in a privacy-preserving manner.\\
\indent Here, we first consider the pagerank \cite{page:stanford99} as an effective graph metric. For the Facebook dataset, we have the average differences between the perturbed pagerank score and the original pagerank score as $0.0016$ and $0.0018$ for $k=5$ and $k=20$ respectively in \system{}. In comparison, the average differences are $0.0019$ and $0.0087$ for $k=5$ and $k=20$ in the approach of Mittal et al. {\system{} preserves the pagerank score of the original graph with up to 4x improvement over previous methods.}
Next, we show the modularity \cite{newman:PNAS06} (computed by the timestamp $t=3$ in the Google+
dataset and the Facebook dataset, respectively) in Table~\ref{modularity}.
We can see that \system{} preserves both the pagerank score and the modularity
of the original graph, while the method
of Mittal et al. degrades such graph analytics especially for larger perturbation parameter $k$
 (recall the visual intuition of \system{} in Fig.~\ref{visualutility}).
\subsection{Summary for Applications of LinkMirage}
\begin{enumerate}[$\bullet$]
\item{\system{} preserves the privacy of users' social contacts while enabling the design of social relationships based applications. Compared to previous methods, \system{} results in significantly lower attack probabilities {(with a factor up to 2)} when applied to anonymous communications and higher resilience to de-anonymization attacks {(with a factor up to 10)} when applied to vertex anonymity systems.}
\item{\system{} even surprisingly improves the Sybil detection performance when applied to the distributed SybilLimit systems.}
\item{\system{} preserves the utility performance for multiple graph analytics applications, such as pagerank score and modularity {with up to 4x improvement}.}
\end{enumerate}

\vspace{-1em}
\section{Utility Analysis}\label{utilitysec}
Following the application analysis in Section~\ref{app}, we aim to develop a general metric
to characterize the utility of the perturbed
graph topologies. Furthermore, we theoretically
analyze the lower bound on utility for \system{},
uncover connections between our utility metric
and structural properties of the graph sequence,
and experimentally analyze our metric using
the real-world Google+ and Facebook datasets.
\vspace{-1em}
\subsection{Metrics}
We aim to formally quantify the utility provided by \system{} to encompass a broader range of applications.
One intuitive global utility metric is the degree
of vertices. It is interesting to find that the expected degree of each node in the
perturbed graph is the same as the original degree
and we defer the proof to Appendix to improve readability.
\begin{mythe}\label{degree_expectation} The expected degree of each
node after perturbation by \system{} is the same as in the
original graph: $\forall v \in V_t, \mathbb{E}(\deg^\prime(v))=\deg(v)$,
where $\deg^\prime(v)$ denotes the degree of vertex $v$ in $G_t^\prime$.
\end{mythe}
\new{\begin{table}[!t]\footnotesize
	\newcommand{\tabincell}[2]{\begin{tabular}{@{}#1@{}}#2\end{tabular}}
	\centering
	\caption{\new{Graph Metrics of the Original and the Perturbed Graphs for the Google+ Dataset.}}
	\begin{tabular}{c|c|c}
	\hline
& Clustering Coefficient & Assortativity Coefficient\\
\hline
Original Graph &0.2612 & -0.0152 \\
 \hline
\system{} $k=2$ &    0.2263 &-0.0185 \\
 \hline
\system{} $k=5$ &    0.1829 &-0.0176  \\
 \hline
\system{} $k=10$ &   0.0864&  -0.0092\\
 \hline
\system{} $k=20$ &   0.0136&  -0.0063\\
\hline
	\end{tabular}
	\label{metrics}
	\end{table}}
\new{To understand the utility in a fine-grained level, we further define our utility metric as} 
\begin{mydef}\label{SUDdefinitionwhole} The Utility Distance (UD)
of a perturbed graph sequence $G_0^{\prime},\cdots,
G_T^{\prime}$ with respect to the original graph sequence $G_0,\cdots, G_T$, and
an application parameter $l$ is defined as 
\begin{equation}\label{SUD}
\begin{aligned}
&\mathrm{UD}(G_0,\cdots G_T,G_0^\prime,\cdots G_T^{\prime},l)\\
\vspace{-0.3em}
&=\frac{1}{T+1}\sum_{t=0}^T \sum_{v\in V_t}\frac{1}{|V_t|}\|P_t^l(v)-(P_t^{\prime})^l(v)\|_{TV}
\end{aligned}
\end{equation}
\end{mydef}
\vspace{-0.5em}
\begin{figure*}[!t]
\renewcommand{\captionfont}{\footnotesize}
\centering
\label{SUDgoogle} 
\includegraphics[width=2.7in,height=1.3in]{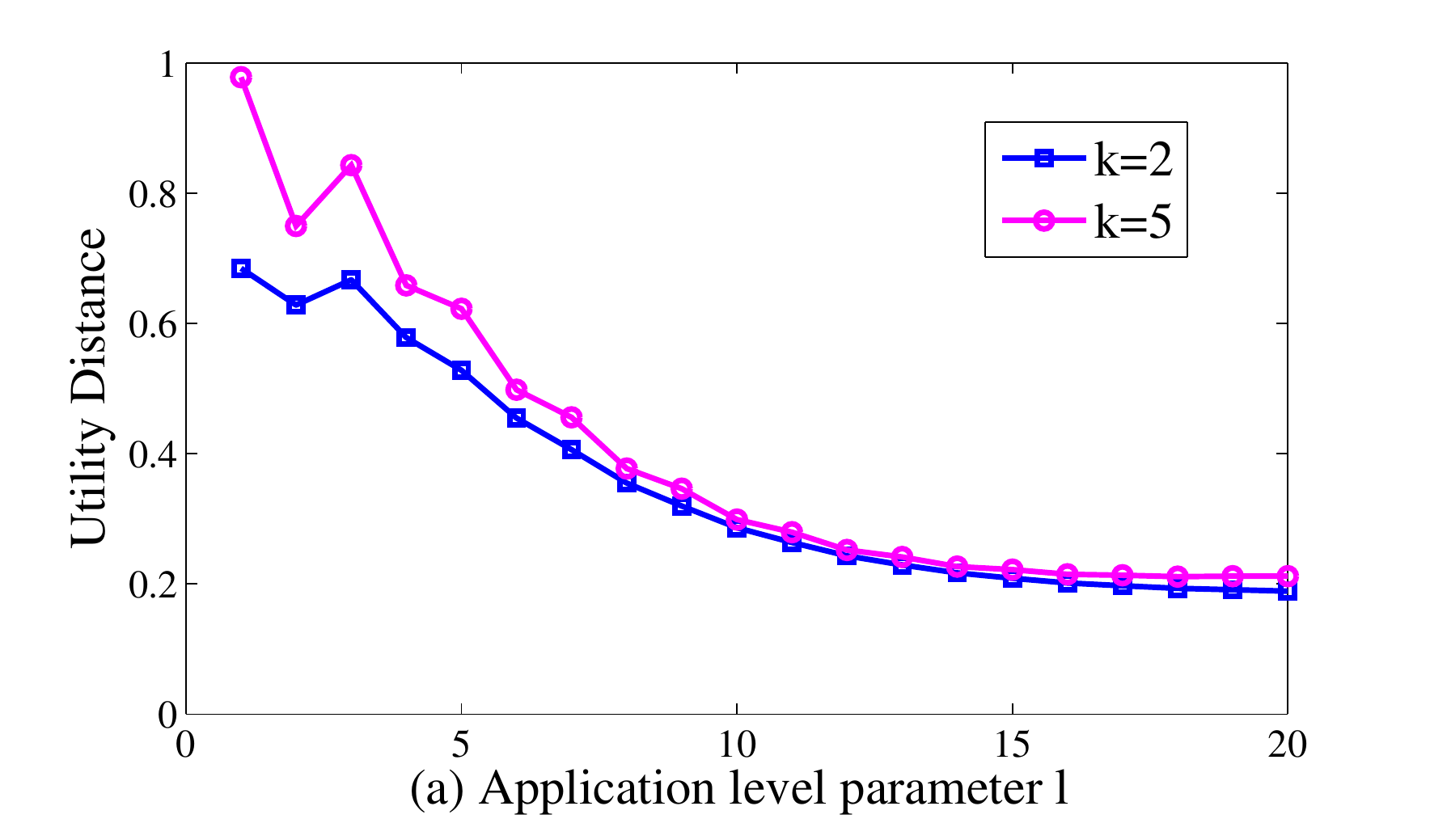}
\hspace{0.5in}
\centering
\label{static_utility_whole} 
\includegraphics[width=2.7in,height=1.3in]{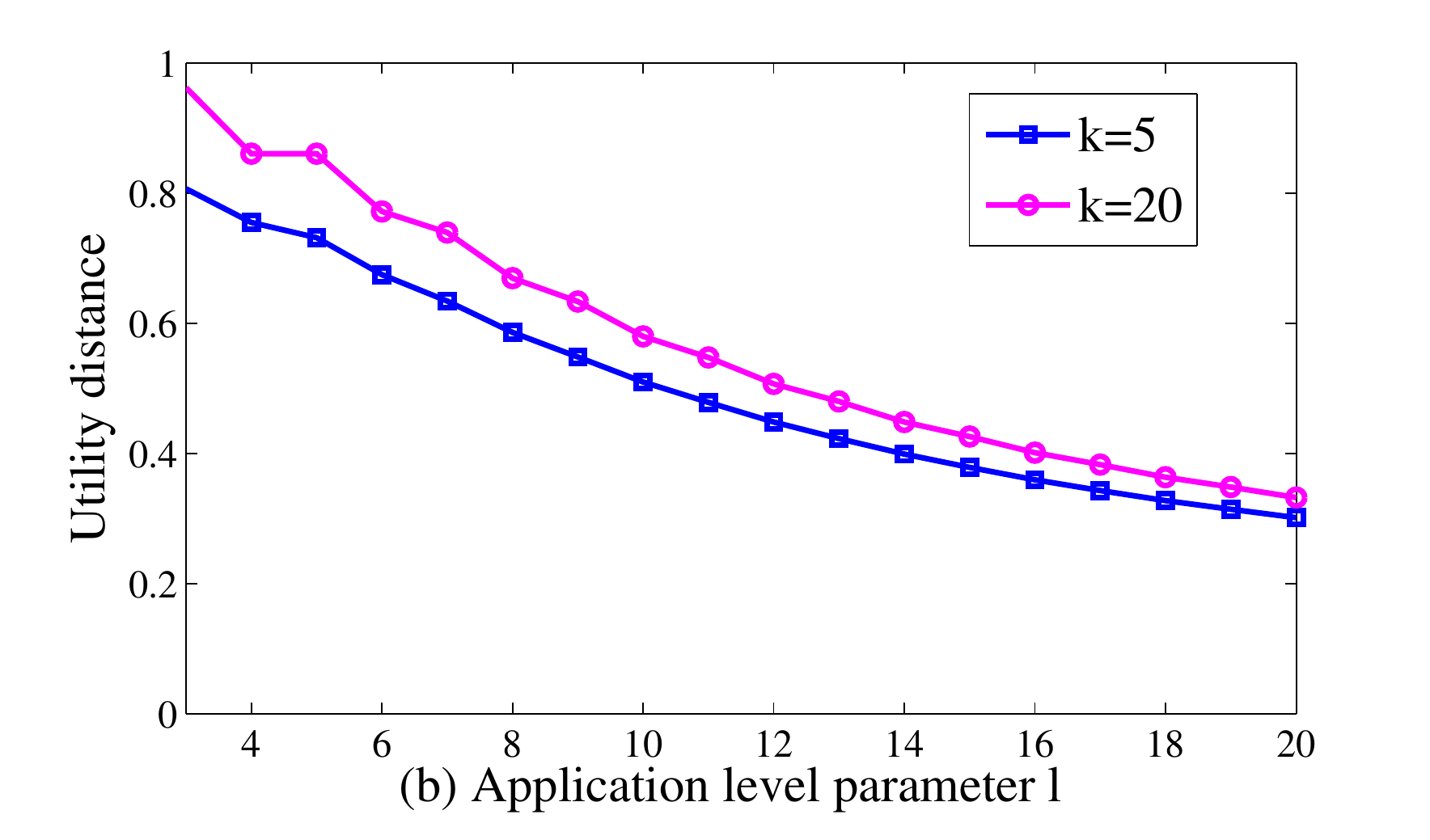}
\vspace{-0.8em}
\caption{(a), (b) show the utility distances using the Google+ dataset and
the Facebook dataset, respectively. Larger perturbation
parameter $k$ results in larger utility distance. Larger application
parameter $l$ decreases the distance, which shows the
effectiveness of \system{} in preserving global community structures.}
\label{SUDTUD} 
\end{figure*}
\indent \new{Our definition for utility distance in Eq.~\ref{SUD} is intuitively reasonable for a broad class of real-world applications, and captures the behavioral
differences of $l$-hop random walks between the original
graphs and the perturbed graphs. We note that
random walks are closely linked to the structural properties
of social networks.
In fact, a lot of social network based security applications
such as Sybil defenses \cite{Yu:IEEES&P08} and anonymity systems \cite{Mittal:piscesNDSS13} directly
perform random walks in their protocols.
The parameter $l$ is application specific;
for applications that require access to fine
grained local structures, such as
recommendation systems \cite{andersen:WWW08}, the value of $l$ should be
small. For other applications that utilize coarse and macro structure of the social graphs, such as Sybil defense mechanisms,
$l$ can be set to a larger value
(typically around 10 in~\cite{Yu:IEEES&P08}). Therefore, this utility metric can quantify the utility performance of \system{} for various applications in a general manner.\\
\indent \new{Note that \system{} is not limited to only preserving the community structure of the original graphs.  We evaluate two representative graph theoretic metrics \emph{clustering coefficient} and \emph{assortativity coefficient}~\cite{gong:IMC12} as listed in Table~\ref{metrics}. We can see that \system{} well preserves such fine-grained structural properties for smaller perturbation parameter $k$. Therefore, the extent to which the utility properties are preserved depends on the perturbation parameter $k$.}}
\vspace{-1em}
\subsection{Relationships with Other Graph Structural Properties}
\new{\indent The mixing time $\tau_\epsilon (G_t)$ measures the time required for the Markov
chain to converge to its stationary distribution, and is defined as
$\tau_{\epsilon}(G_t)=\min_{r}\max_{v}(r\big| |P_t^r(v)-\pi_t|_{\mathrm{TV}}<\epsilon)$.
Based on the Perron-Frobenius theory, we denote the eigenvalues of $P_t$ as
$1=\mu_1(G_t)\ge \mu_2(G_t)\ge \cdots \mu_{|V_t|}(G_t)\ge -1$. The convergence rate of the Markov
chain to $\pi_t$ is determined by the second largest eigenvalue modulus (SLEM) as
$\mu(G_t)=\max \left(\mu_2(G_t),-\mu_{|V_t|}(G_t)\right)$.\\
\indent Since our utility distance is defined by using the transition probability matrix $P_t$,
this metric can be proved to be closely related to structural properties of the graphs,
as shown in Theorem~\ref{mixingtime} and Theorem~\ref{SLEM}.
\vspace{-0.5em}
\begin{mythe}\label{mixingtime} Let us denote the utility distance between the perturbed graph
 $G_t^{\prime}$ and the original graph $G_t$ by $\mathrm{UD}(G_t, G_t^{\prime},l)$, then we have $\tau_{G_t^{\prime}}\left(\mathrm{UD}\left(G_t,G^{\prime}_t,\tau_{G_t}(\epsilon)\right)-\epsilon\right)\ge \tau_{G_t}(\epsilon)$.
\end{mythe}
\vspace{-0.5em}
\begin{mythe}\label{SLEM} Let us denote the second largest eigenvalue modulus (SLEM)
 of transition probability matrix $P_t$ of graph $G_t$ as $\mu_{G_t}$. We can bound the SLEM of a perturbed graph
  $G_t^{\prime}$ using the mixing time of the original graph, and the utility distance between the graphs as $\mu_{G_t^\prime}\ge1-\frac{\log n +\log {\frac{1}{\mathrm{UD} (G_t,G_t^\prime ,\tau_{G_t} (\epsilon))-\epsilon}}}{\tau_{G_t} (\epsilon )}$.
\end{mythe}}
\vspace{-1em}
\subsection{Upper Bound of Utility Distance}
\vspace{-1em}
\system{} aims to limit the
degradation of link privacy over time.
Usually, mechanisms that preserve privacy trade-off application
utility. In the following, we will theoretically derive an upper bound
on the utility distance for our algorithm. This corresponds
to a lower bound on utility that \system{} is guaranteed to provide.
\vspace{-1em}
\begin{mythe}\label{SUDupperbound} The  utility distance of \system{} is upper bounded by $2l$ times the sum of the utility distance of each community $\epsilon$ and the ratio cut $\delta_t$ for each $G_t$, i.e.
\begin{equation}
\mathrm{UD}(G_0,\cdots G_T,G_0^\prime,\cdots G_T^{\prime},l)
\le\frac{1}{T+1}\sum_{t=0}^T 2l(\epsilon+\delta_t)
\end{equation}
\end{mythe}
\vspace{-1em}
where $\delta_t$ denotes the number of inter-community links
over the number of vertices, and each community $C_{k(t)}$ within $G_{t}$ satisfies $\|C_{k(t)}-C_{k(t)}^{\prime}\|_{TV}\le\epsilon$. {We defer the proofs to the Appendix to improve readability}.\\
\indent Note that an upper bound on utility distance corresponds
to a lower bound on utility of our algorithm.
While better privacy usually requires adding more noise
to the original sequence to obtain the perturbed sequence,
thus we can see that LinkMirage is guaranteed to provide
a minimum level of utility performance.\\
\indent In the derivation process, we do not take specific
evolutionary pattern such as the overlapped ratio
into consideration, therefore our theoretical
 upper bound is rather loose. 
Next, we will show
that in practice, \system{}
achieves smaller utility distance (higher utility)
than the baseline approach of independent static
perturbations.
\vspace{-1em}
\subsection{Utility Experiments Analysis}
\vspace{-1em}
Fig.~\ref{SUDTUD}(a)(b) depict the utility
distance for the Google+ and the Facebook graph sequences, for varying perturbation degree $k$
and the application level parameter $l$.
We can also see that as $k$ increases, the distance metric
increases. This is natural since additional noise 
increase
the distance between probability distributions computed from
the original and the perturbed graph series. As the application parameter
$l$ increases, the distance metric decreases. This illustrates that
\system{} is more suited for security applications
that rely on macro structures, as opposed to applications
that require exact information about one or two hop neighborhoods. \final{Furthermore, our experimental results in Figure~\ref{pisces} and Table~\ref{modularity} also demonstrate the utility advantage of our \system{} over the approach of Mittal et al.~\cite{Mittal:NDSS13} in real world applications.}

\vspace{-1em}
\section{Related Work}
\noindent{\bf{Privacy with labeled vertices}} An important 
thread of research aims to preserve link privacy between
labeled vertices by obfuscating the edges,
i.e., by adding /deleting edges~\cite{Hay:csfacultypublication07,Mittal:NDSS13,Ying:SIAM08}.
These methods aim to randomize the structure of the social graph, while differing in
the manner of adding noise.  Hay et al. \cite{Hay:csfacultypublication07} perturb the
graph by applying a sequence of $r$ edge deletions and $r$ edge insertions.
The deleted edges are uniformly selected from the existing edges in the original graph while the added edges are uniformly selected
from the non-existing edges. However, neither the edge deletions nor edge insertions take any structural properties of the graph into consideration.
Ying and Wu \cite{Ying:SIAM08} proposed a new
perturbation method for preserving spectral properties, without analyzing its
privacy performance.\\
\indent Mittal et al. proposed a perturbation method in \cite{Mittal:NDSS13},
which serves as the foundation for our algorithm.
Their method deletes \emph{all} edges in the original graph, and replaces each edge
with a fake edge that is sampled based on the structural properties of the
graph. In particular, random walks are performed on the original graph
to sample fake edges. 
{As compared to the methods of Hay et al.~\cite{Hay:csfacultypublication07}
and Mittal et al.~\cite{Mittal:NDSS13}, \system{} provides up to 3x privacy improvement for static social 
graphs and up to 10x privacy improvement for dynamic social graphs.} \\
\indent Another line of research aims to preserve link privacy \cite{Hay:VLDB08} \cite{Zheleva:KDD08} by 
aggregating the vertices and edges into super vertices. Therefore, the privacy of links
within each super vertex is naturally protected. However, such approaches do not
permit fine grained utilization of graph properties, making it difficult to be applied to applications
such as social network based anonymous communication and Sybil defenses. \\
{\bf{Privacy with unlabeled vertices}} While the focus of our paper is on preserving link privacy 
in context of labeled vertices, an orthogonal line of research aims to provide privacy in the context of 
unlabeled vertices (vertex privacy) ~\cite{liu:sigmod08,sala:imc11,beato:percom13}. Liu et al. \cite{liu:sigmod08} proposed $k$-anonymity to anonymize unlabeled vertices by placing at least $k$ vertices at an equivalent level. Differential privacy provides a theoretical framework for perturbing aggregate information, and Sala et al. \cite{sala:imc11} leveraged differential privacy to privately publish social graphs with unlabeled vertices. We note that \system{} can also provide a 
foundation for preserving vertex privacy as stated in Section~\ref{vertex}. \newnew{Shokri et al. \cite{shokri2015privacy} addresses the privacy-utility trade-off by using game theory, which does not consider the temporal scenario.} \\
\indent We further consider anonymity in temporal graphs with unlabeled vertices. \emph{The time series data should be seriously considered}, since the adversaries can combine multiple published graph to launch enhanced attacks for inferring more information. \cite{tai:ICDE11,ding:globalcom11,Bhagat:WWW10} explored privacy degradation in vertex privacy schemes due 
to the release of multiple graph snapshots. These observations motivate our work, even though we focus on 
labeled vertices.\\
{\bf{De-anonymization}} In recent years, the security community has proposed a number of 
sophisticated attacks for de-anonymizing social graphs~\cite{Narayaran:S&P09,srivatsa:CCS12,nilizadeh:CCS14,ji:CCS14}. 
While most of these attacks are not applicable to link privacy mechanisms (their focus is on vertex privacy), they illustrate the importance 
of considering adversaries with prior information about the social graph\footnote{Burattin et al~\cite{burattin:arxiv14} 
exploited inadvertent information leaks via Facebook's graph API to de-anonymize social links; 
Facebook's new graph API (v2.0) features stringent privacy controls as a countermeasure.}. We perform a rigorous privacy analysis 
of \system{} (Section~\ref{privacy}) by considering a worst-case (strongest) adversary that knows the entire social graph except one link, 
and show that even such an adversary is limited in its inference capability. 


\vspace{-1em}
\section{Discussion}
{Privacy Utility Tradeoffs: \system{} mediates privacy-preserving access to users' social 
relationships. In our privacy analysis, we consider the worst-case 
adversary who knows the entire social link information except one link, which conservatively 
demonstrates the superiority of our algorithm over the state-of-the-art approaches. 
\system{} benefits many applications that depend on graph-theoretic properties of the social graph (as opposed to the exact set of edges). 
This also includes recommendation systems and E-commerce applications.}\\
\indent{Broad Applicability: While our theoretical analysis of \system{} relies on undirected links, the obfuscation algorithm 
itself can be generally applied to directed social networks.  Furthermore, our underlying techniques have broad 
applicability to domains beyond social networks, 
including communication networks and web graphs. 
}
\vspace{-1em}
\section{Conclusion}
\system{} effectively mediates privacy-preserving access to users' social relationships, {since 1) \system{} preserves key structural properties in the
social topology while anonymizing intra-community and inter-community links;} 2) \system{} provides rigorous guarantees for the anti-inference privacy, indistinguishability and anti-aggregation privacy, in order to defend against sophisticated threat models for both static and temporal graph topologies; 3) \system{} significantly
outperforms baseline static techniques in terms of both
link privacy and utility, which have been verified both theoretically and
experimentally using real-world Facebook dataset (with
\facebook{} links) and the large-scale Google+ dataset
(with 940M links).  \system{} enables the deployment of real-world social relationship based applications 
such as graph analytic, anonymity systems, and Sybil defenses while preserving the privacy of users' social relationships.
\vspace{-1em}

\bibliographystyle{IEEEtranS}
\bibliography{mybib} 
\vspace{-1em}
\footnotesize
\section{Appendix}
\noindent{\bf{A. Proof of the Upper Bound of Anti-aggregation Privacy}}
$\|P_t^k-P_t^\prime\|_{\mathrm{TV}}=\|P_t^k-\hat P_t^k\|_{\mathrm{TV}}\le\frac{1}{2|V_t|}\sum\nolimits_{v=1}^{|V_t|} \|P_t(v)P_t^{k-1}-\hat P_t(v) P_t^{k-1}\|_1+\frac{1}{2|V_t|}\sum\nolimits_{v=1}^{|V_t|}\|\hat P_t(v)P_t^{k-1}-\hat P_t(v) \hat P_t^{k-1}\|_1=\|P_t-\hat P_t\|_{\mathrm{TV}}+\|P_t^{k-1}-\hat P_t^{k-1}\|_{\mathrm{TV}}\le k \|P_t-\hat P_t\|_{\mathrm{TV}}.$
\noindent {\bf{B. Relationships with Differential Privacy}}
When considering differential privacy for a time series of graph sequence $\{G_i\}_{i=0}^t$, we have
$f(D)=P(\{G^\prime_i\}_{i=0}^t|\{\widetilde{G}_i(L_t)\}_{i=0}^t, L_t=1), f(D^\prime)=P(\{G^\prime_i\}_{i=0}^t|\{\widetilde{G}_i(L_t)\}_{i=0}^t, L_t=0)$. For a good privacy performance, we need 
$
P(\{G^\prime_i\}_{i=0}^t|\{\widetilde{G}_i(L_t)\}_{i=0}^t,L_t=1)\approx P(\{G^\prime_i\}_{i=0}^t|\{\widetilde{G}_i(L_t)\}_{i=0}^t,L_t=0)$.
\indent Since the probability of $\{G^\prime_i\}_{i=0}^t$
given $\{\widetilde{G}_i(L_t)\}_{i=0}^t$ as $P(\{G^\prime_i\}_{i=0}^t|\{\widetilde{G}_i(L_t)\}_{i=0}^t)
=P(\{G^\prime_i\}_{i=0}^t|\{\widetilde{G}_i(L_t)\}_{i=0}^t, L_t=1)P(L_t=1|\{\widetilde{G}_i(L_t)\}_{i=0}^t)
+P(\{G^\prime_i\}_{i=0}^t|\{\widetilde{G}_i(L_t)\}_{i=0}^t, L_t=0)P(L_t=0|\{\widetilde{G}_i(L_t)\}_{i=0}^t$, it is easy to see that if the condition for a good privacy performance holds, we have $P(L_t|\{{G}^\prime_i\}_{i=0}^t,\{\widetilde{G}_i(L_t)\}_{i=0}^t)=\frac{P(\{{G}^\prime_i\}_{i=0}^t|\{\widetilde{G}_i(L_t)\}_{i=0}^t, L_t)\times P(L_t|\{\widetilde{G}_i(L_t)\}_{i=0}^t)}{P(\{G^\prime_i\}_{i=0}^t|\{\widetilde{G}_i(L_t)\}_{i=0}^t)}\approx P(L_t|\{\widetilde{G}_i(L_t)\}_{i=0}^t)$, 
which is the same as in Definition \ref{bayesianprivacy} and
means that the posterior probability is similar to the prior
probability, i.e., the adversary is bounded in the information it can learn from the perturbed graphs.\\
\noindent{\bf{C. Proof of Theorem \ref{degree_expectation}: Expectation of Perturbed Degree}}
According to Theorem 3 in \cite{Mittal:NDSS13}, we have
$
\mathbb{E}(\deg_{com}^{\prime}(v))=\deg(v)
$, 
where $\deg_{com}^{\prime}(v)$ denotes the degree of $v$ after perturbation within community.
Then we consider the random perturbation for inter-community subgraphs. Since the probability for an edge to be chosen is $\frac{\deg(v_a(i))\deg(v_b(j))|v_a|}{|E_{ab}|(|v_a|+|v_b|)}$, the expected degree after inter-community perturbation satisfies
$\mathbb{E}(\deg^{\prime}_{inter}(v_a(i)))=\sum_j{\frac{\deg(v_a(i))\deg(v_b(j))(|v_a|+|v_b|)}{|E_{ab}|(|v_a|+|v_b|)}}\\
=\deg(v_a(i))$. Combining with the expectations under static scenario, we have
$
\mathbb{E}(\deg^\prime(v))=\deg(v).
$
%
\\ \noindent{\bf{D. Proof of the Upper Bound for the Utility Distance}}
We first introduce some notations and concepts.
We consider two perturbation methods in the derivation process below.
The first method is our dynamic perturbation method, which
takes the graph evolution into consideration. The second method is the intermediate method, where we only implement
dynamic clustering without selective perturbation.
That is to say, we cluster $G_t$, then perturb each community by the static method and
each inter-community subgraphs by randomly connecting the marginal nodes, independently.
We denote the perturbed graphs corresponding to the dynamic,
the intermediate method by
$G_t^{\prime}, G_t^{\prime,i}$ respectively. Similarly,
we denote the perturbed TPM for the two approaches by
$P_t^{\prime},P_t^{\prime,i}$.
To simplify the derivation process, we partition
the proof into two stages. In the first stage, we derive the UD upper bound for the intermediate perturbation method.
In the second stage, we derive the relationship between $G_{t}^{\prime,i}$ and $G_{t}^{\prime}$.
Results from the two stages can be combined to find the upper bound for the utility distance of \system{}.
Denoting the communities as $C_{1}, C_{2}, \cdot, C_{K_t}$ and the inter-community subgraphs as $C_{12}, C_{13}, \cdots$, we have
\begin{equation}
\begin{aligned}
&\|P_t-{P_{t}}^{\prime,i}\|_{\mathrm{TV}}\\
&=\left\| \begin{array}{ccc}
P_{t(1,1)}-P_{t(1,1)}^{\prime} & \ldots & P_{t(1,K_t)}-P_{t(1,K_t)}^{\prime}\\
P_{t(2,1)}-P_{t(2,1)}^{\prime}  &\ldots & P_{t(2,K_t)}-P_{t(2,K_t)}^{\prime}\\
\vdots & \vdots  & \ddots\\
P_{t(K_t,1)}-P_{t(K_t,1)}^{\prime} & \ldots & P_{t(K_t,K_t)}-P_{t(K_t,K_t)}^{\prime}
\end{array} \right\|_{\mathrm{TV}}\\
&=\frac{1}{|V_t|}{\sum\nolimits_{k=1}^{K_t}|V_t(k)| \|P_{t(k,k)}-P_{t(k,k)}^{\prime}\|_{\mathrm{TV}}}\\
&+\frac{1}{|V_t|}{\sum\nolimits_{k,j=1,k\ne j}^{K_t}|E_t{(k,j)}| \|P_{t(k,j)}-P_{t(k,j)}^{\prime}\|_{\mathrm{TV}}}\\
&\le\epsilon+\delta_t\\
\end{aligned}
\end{equation}
Here, $\delta_t$ is the ratio cut of the graph \cite{Jesus:EPIA09}, and
$\delta_t={|E_{t\mathrm{-in}}|}/{|V_t|}=\sum\nolimits_{k,j=1, k\ne j}^{K_t}{|E_t(k,j)|}/{|V_t|}$.
For arbitrary matrix $P$ and $Q$, we have
$
\|P^l-Q^l\|_{\mathrm{TV}}\le l\|P-Q\|_{\mathrm{TV}}
$.
Combining the above results, we have
\begin{equation}\label{upperboundintermediate}
\mathrm{UD}(G_t,G_{t}^{\prime,i},l)\le l \|P_t-P_{t}^{\prime,i}\|_{\mathrm{TV}}\le l\left(\epsilon+\delta_t\right)
\end{equation}
Then, we generalize the utility analysis of intermediate perturbation to our dynamic perturbation.
Assume that there are $K_{t}^{\mathrm{c}}$ out of $K_t$ clusters that are considered as changed, which would be perturbed independently, and $K_{t}^{\mathrm{u}}$ out of $K_t$ clusters are considered as unchanged, i.e., their perturbation would follow the perturbation manner in $G_{t-1}^{\prime}$. To simplify derivation, we use $P_{t(k)}$ instead of $P_{t(k,k)}$ to represent the TPM of the $k$-th community. Then, we have
\begin{equation}
\begin{aligned}
&\mathrm{UD}(G_t, G_t^{\prime},1)=\|P_t-P_t^\prime\|_{\mathrm{TV}}\\
&\le\frac{\sum_{k=1}^{K_{t}^{\mathrm{c}}}\|P_{t(k)}-P_{t(k)}^\prime\|_{\mathrm{TV}}+\sum_{j=1}^{K_{t}^{\mathrm{u}}}\|P_{t(j)}-P_{t(j)}^\prime\|_{\mathrm{TV}}}{|K_t|}+\delta_t\\
& \leq \frac{1}{|K_t|}\left(\sum\limits_{k=1}^{K_{t}^{\mathrm{c}}}\|P_{t(k)}-P_{t(k)}^\prime\|_{\mathrm{TV}} +\sum\limits_{j=1}^{K_{t}^{\mathrm{u}}}\left(\|P_{t(j)}-P_{(t-1)(j)}\|_{\mathrm{TV}}\right.\right.\\
&\left.\left.+\|P_{(t-1)(j)}-P_{(t-1)(j)}^\prime\|_{\mathrm{TV}}+\|P_{(t-1)(j)}^\prime-P_{t(j)}^\prime\|_{\mathrm{TV}}\right)\right)+\delta_t\\
&= \frac{\sum_{k=1}^{K_t}\|P_{t(k)}-P_t^{\prime}{(k)}\|_{\mathrm{TV}}}{|K_t|}+\frac{|K_{t}^{\mathrm{u}}|\epsilon_0}{|K_t|}+\delta_t\\
&\le \mathrm{UD}(G_t, G_t^{\prime,i},1)+\epsilon+\delta_t
\end{aligned}
\end{equation}
where $\epsilon_0$ denotes the threshold to classify a community as \textit{changed} or \textit{unchanged}. The last inequality comes from the fact that $\epsilon_0\le \epsilon$. Then, we can prove 
$\mathrm{UD}\left(G_t,G_t^{\prime},l\right)=\|P_t^l-(P_t^\prime)^l\|_{\mathrm{TV}}\leq l \|P_t-P_t^\prime\|_{\mathrm{TV}}\leq l \|P_t-P_{t}^{\prime,i}\|_{\mathrm{TV}}+l(\epsilon+\delta_t)=2l\left(\epsilon+\delta_t\right)$ and $
\mathrm{UD}(G_0,\cdots G_T,G_0^\prime,\cdots G_T^{\prime},l)
\le\frac{1}{T+1}\sum_{t=0}^T 2l(\epsilon+\delta_t).$\\
\noindent {\bf{E. Proof for Relating Utility Distance with Structural Metrics}}
From the definition of total variation distance, we have $\|P_v^r(G^{\prime}_t)-\pi\|_{TV}+\|P_v^r(G_t)-\pi\|_{TV}\ge \|P_v^r(G^{\prime}_t)-P_v^t(G)\|_{TV}$. Taking the maximum over all vertices, we have $\max\|P_v^r(G^{\prime}_t)+\pi\|_{TV}+\max\|P_v^r(G_t)-\pi\|_{TV}\ge \max \|P_v^r(G^{\prime}_t)-P_v^t(G)\|_{TV}$. Therefore, for $t\ge \tau_{G}(\epsilon)$,
$\max\|P_v^r(G^{\prime}_t)-\pi\|_{TV}\ge \max \|P_v^r(G^{\prime}_t)-P_v^t(G)\|_{TV}+\max\|P_v^r(G_t)-\pi\|_{TV}\ge\frac{\sum_{v=1}^{|V_t|}\|P_v^r(G^{\prime}_t)-P_v^t(G)\|_{TV}-\pi\|_{TV}}{|V_t|}-\epsilon=\mathrm{UD}(G_t,G_t^{\prime},\tau_{G}(\epsilon))-\epsilon$. Then, we have
$
\tau_{G_t^{\prime}}\left(\mathrm{UD}\left(G_t,G^{\prime}_t,\tau_{G_t}(\epsilon)\right)-\epsilon\right)\ge \tau_{G_t}(\epsilon)
$. It is known that the second largest eigenvalue modulus is related to the mixing time of the graph as $
\tau_{G^\prime_t}(\epsilon) \leq \frac{\log n +\log {\frac{1}{\epsilon}}}{1-\mu_{G^\prime_t}}$. From this relationship, we can bound the SLEM in terms of the mixing time as $1-\frac{\log n +\log {(\frac{1}{\epsilon}})}{\tau_{G^\prime_t}} \leq \mu_{G^\prime_t}$. Replacing $\epsilon$ with $\mathrm{UD}(G_t,G_t^\prime,\tau_{G_t} (\epsilon))-\epsilon$, we have
$1-\frac{\log n +\log {\frac{1}{\mathrm{UD} (G_t,G_t^\prime ,\tau_{G_t} (\epsilon))-\epsilon}}}{\tau_{G_t^\prime} (\mathrm{UD} (G_t,G_t^\prime ,\tau_{G_t} (\epsilon)))-\epsilon} \leq \mu_{G_t^\prime}$. Finally, we leverage $\tau_{G_t^\prime} (\mathrm{UD} (G_t,G_t^\prime,\tau_{G_t} (\epsilon)-\epsilon))\geq \tau_{G_t} (\epsilon)$ in the above equation, to obtain $
\mu_{G_t^\prime}\ge1-\frac{\log n +\log {\frac{1}{\mathrm{UD} (G_t,G_t^\prime ,\tau_{G_t} (\epsilon))-\epsilon}}}{\tau_{G_t} (\epsilon )}$.

\end{document}